\documentclass[a4paper,twocolumn,11pt,accepted=2024-02-17]{quantumarticle}
\pdfoutput=1
\usepackage[utf8]{inputenc}
\usepackage[english]{babel}
\usepackage[T1]{fontenc}
\usepackage{graphicx}
\usepackage{dcolumn}
\usepackage{bm}
\usepackage{multirow}
\usepackage{amssymb}
\usepackage{amsfonts}
\usepackage{amsmath}
\usepackage{amsthm}
\usepackage{enumerate}
\usepackage{color}
\usepackage[dvipsnames]{xcolor}
\usepackage{marvosym}
\usepackage{qcircuit}

\usepackage{mathtools}

\usepackage{hyperref}

\usepackage{soul,xcolor}

\usepackage{tikz}
\usetikzlibrary{matrix,positioning,decorations.pathreplacing}

\usepackage[normalem]{ulem}
\newcommand\pst{\bgroup\markoverwith{\textcolor{purple}{\rule[0.5ex]{2pt}{0.4pt}}}\ULon}
\newcommand\bst{\bgroup\markoverwith{\textcolor{blue}{\rule[0.5ex]{2pt}{0.4pt}}}\ULon}

\newtheorem{theorem}{Theorem}
\newtheorem{corollary}[theorem]{Corollary}
\newtheorem{lemma}[theorem]{Lemma}
\newtheorem{proposition}[theorem]{Proposition}

\newtheorem{observation}[theorem]{Observation}

\newcommand{\newreptheorem}[2]{\newtheorem*{rep@#1}{\rep@title}\newenvironment{rep#1}[1]{\def\rep@title{#2 \ref*{##1}}\begin{rep@#1}}{\end{rep@#1}}}

\newreptheorem{theorem}{Theorem}
\newreptheorem{proposition}{Proposition}
\newreptheorem{observation}{Observation}
\newreptheorem{corollary}{Corollary}

\newcommand{\bra}[1]{\langle #1|}
\newcommand{\ket}[1]{|#1\rangle}

\newcommand{\be}{\begin{equation}}
\newcommand{\ee}{\end{equation}}
\newcommand{\bea}{\begin{eqnarray}}
\newcommand{\eea}{\end{eqnarray}}

\newcommand{\kommentar}[1]{}

\newcommand{\identity}{\mathbbm{1}}

\newcommand{\forget}[1]{}

\begin{document}

\title{Identifying families of multipartite states with non-trivial local entanglement transformations}

\author{Nicky Kai Hong Li}
\affiliation{Institute for Theoretical Physics, University of Innsbruck, Technikerstr. 21A, 6020 Innsbruck, Austria}\affiliation{Department of Physics, QAA, Technical University of Munich, James-Franck-Str.\;1, D-85748 Garching, Germany}\affiliation{Current address: Atominstitut, Technische Universit\"{a}t Wien, Stadionallee 2, 1020 Vienna, Austria}
\orcid{0000-0002-4087-4744}
\author{Cornelia Spee}
\affiliation{Institute for Theoretical Physics, University of Innsbruck, Technikerstr. 21A, 6020 Innsbruck, Austria}
\author{Martin Hebenstreit}
\affiliation{Institute for Theoretical Physics, University of Innsbruck, Technikerstr. 21A, 6020 Innsbruck, Austria}
\orcid{0000-0002-5841-7082}
\author{Julio I. de Vicente}
\affiliation{Departamento de Matemáticas, Universidad Carlos III de Madrid, Avda. de la Universidad 30, E-28911, Leganés (Madrid), Spain}\affiliation{Instituto de Ciencias Matemáticas (ICMAT), E-28049 Madrid, Spain}
\orcid{0000-0002-6508-5709}
\author{Barbara Kraus}
\affiliation{Institute for Theoretical Physics, University of Innsbruck, Technikerstr. 21A,  6020 Innsbruck, Austria}\affiliation{Department of Physics, QAA, Technical University of Munich, James-Franck-Str.\;1, D-85748 Garching, Germany}
\orcid{0000-0001-7246-6385}
\maketitle

\begin{abstract}
The study of state transformations by spatially separated parties with local operations assisted by classical communication (LOCC) plays a crucial role in entanglement theory and its applications in quantum information processing. Transformations of this type among pure bipartite states were characterized long ago and have a revealing theoretical structure. However, it turns out that generic fully entangled pure multipartite states cannot be obtained from nor transformed to any inequivalent fully entangled state under LOCC. States with this property are referred to as isolated. Nevertheless, multipartite states are classified into families, the so-called SLOCC classes, which possess very different properties. Thus, the above result does not forbid the existence of particular SLOCC classes that are free of isolation, and therefore, display a rich structure regarding LOCC convertibility. In fact, it is known that the celebrated $n$-qubit GHZ and W states give particular examples of such classes and in this work, we investigate this question in general. One of our main results is to show that the SLOCC class of the 3-qutrit totally antisymmetric state is isolation-free as well. Actually, all states in this class can be converted to inequivalent states by LOCC protocols with just one round of classical communication (as in the GHZ and W cases). Thus, we consider next whether there are other classes with this property and we find a large set of negative answers. Indeed, we prove weak isolation (i.e., states that cannot be obtained with finite-round LOCC nor transformed by one-round LOCC) for very general classes, including all SLOCC families with compact stabilizers and many with non-compact stabilizers, such as the classes corresponding to the $n$-qunit totally antisymmetric states for $n\geq4$. Finally, given the pleasant feature found in the family corresponding to the 3-qutrit totally antisymmetric state, we explore in more detail the structure induced by LOCC and the entanglement properties within this class.
\end{abstract}

\section{Introduction}

Many quantum technologies that demonstrate advantages over their classical counterparts rely on multipartite entanglement \cite{EntQKD,Gottesman,SecretSharing,MBC,MetrologyGHZ,ByzantineGHZ}. Moreover, in the last years, entanglement theory tools have proven to be instrumental in the description of interactions and correlations in many body systems \cite{reviewmps}. Thus, the resource theory of entanglement is one of the building blocks of quantum information theory \cite{Horodeckis, ResourceTheory}. It aims at characterizing, classifying, and quantifying entanglement, providing protocols to harness this resource, and studying how efficient these manipulations can be. The free operations in this theory are \textit{local operations assisted by classical communication} (LOCC), which naturally describe state manipulation protocols carried out by multiple spatially separated parties. Therefore, the study of state transformations under LOCC provides practical protocols for quantum information processing. More importantly, it plays a crucial role in the foundations of entanglement theory. Since LOCC operations are assumed to be freely implementable, a state cannot be less useful than any state that can be obtained from it by this kind of transformation. Hence, LOCC convertibility defines a meaningful partial order for the entanglement contained in quantum states and defines a basic monotonicity condition that every entanglement measure must satisfy. A simple example of an LOCC map is local unitary (LU) transformations. Given that this is actually an invertible map, all states related by LU transformations are regarded as equivalent in entanglement theory.

For bipartite pure states, LOCC transformability is completely determined by a majorization relation between the Schmidt coefficients of the two states \cite{Nielsen}. This is a seminal result in the field. Among other reasons, this singles out a maximally entangled state which must be the optimal resource for any task to be performed under the LOCC constraint. Unfortunately, the situation in the multipartite case is much more involved. First, it is known that the pure-state space is partitioned into different classes containing all states that can be interconverted with non-zero probability, i.e., by stochastic LOCC (SLOCC) \cite{3qubitSLOCC}. Therefore, no LOCC transformation can exist between any two fully entangled pure states of the same local dimensions in different SLOCC classes. To make things more complicated, except for 2-qudit \cite{Nielsen} and 3-qubit pure states \cite{3qubitSLOCC}, in general, there exist infinitely many different SLOCC classes for other number of parties and local dimensions \cite{4qubitSLOCC, footnote:infSLOCC}. In addition to this, SLOCC only defines an equivalence relation and one is still left with the question of characterizing LOCC convertibility within each SLOCC class in order to obtain a meaningful ordering of states. Interestingly, general tools based on the stabilizer of the SLOCC class (see definition below) have been provided \cite{GourWallach,SEPeqn, ProbStep1, ProbStep2}, which, in principle, allow us to characterize LOCC transformability \cite{MESpaper}. However, a second and more drastic difference between the multipartite and the bipartite cases arises. It turns out that for most values of the number of parties and local dimension, generic SLOCC classes have a trivial stabilizer \cite{Isolation1,Isolation2}. From this, it follows that almost every fully entangled pure multipartite state is isolated \cite{Isolation1,Isolation2}, i.e., it cannot be obtained from nor transformed to any LU-inequivalent fully entangled state of the same local dimensions by LOCC operations. Therefore, except for very specific configurations, entanglement is incomparable among generic states. In fact, this is the case in $(\mathbb{C}^d)^{\otimes n}$ for $n>4$ parties with local dimension $d=2$ and for $d>2$ and $n>3$ \cite{Isolation1,Isolation2}. However, it remains unknown whether this property extends to the unbalanced case, i.e., for states such that the local dimensions for each party are not all the same.

The above notwithstanding, in the perspective of identifying the most useful states for achieving certain tasks that require entanglement, it is still possible to find SLOCC classes with a non-trivial stabilizer that are of measure zero with respect to the whole Hilbert space and contain non-isolated states. Actually, the possible LOCC transformations among the states in the SLOCC classes of the $n$-qubit GHZ and W states were characterized respectively in Refs.\;\cite{turgutGHZ} and \cite{turgutW}. It turns out that these two classes are isolation-free (i.e., no state therein is isolated). In the last few years, the aforementioned techniques have enabled the study of LOCC convertibility in SLOCC classes with a non-trivial stabilizer such as the ones that contain generic \cite{MESpaper} and non-generic \cite{4qubitMES} 4-qubit states, generic 3-qutrit states \cite{locc3qutrit}, graph states and general stabilizer states \cite{MatthiasStab}, matrix product states \cite{loccmps1,loccmps2}, permutation-symmetric states of an arbitrary number of parties and local dimension \cite{OurSymmPaper}, and states relevant for quantum networks \cite{yamasaki, NonTrivial4round}. Remarkably, no further isolation-free SLOCC class has been found with the exception of the particular quantum networks studied in Ref.\;\cite{yamasaki}, which only occur in the unbalanced case. Moreover, the results of Refs.\;\cite{MESpaper} and \cite{locc3qutrit} show that generic 4-qubit and 3-qutrit states are isolated even though the stabilizer is not trivial in these cases. Thus, despite this active line of investigation, the GHZ and W SLOCC classes stand out as the unique instances with such feature for the case of all subsystems having equal dimension. Interestingly, this property seems to go hand in hand with relevant applications as these states appear in, e.g., quantum communication \cite{SecretSharing,ByzantineGHZ,QCommW}, quantum metrology \cite{MetrologyGHZ}, and entanglement robustness against loss \cite{RobustLossW}. Therefore, the underlying question that motivates this work is: being aware that this phenomenon cannot hold generically for almost all balanced configurations, are there more particular SLOCC classes that are isolation-free in addition to the few preexisting examples? Our main motivation is that, as explained above, the study of LOCC protocols and, in particular, finding families of states with non-trivial convertibility properties are questions of fundamental interest in entanglement theory. In addition to this, our hope is that, on the analogy of the aforementioned properties of the GHZ and W states, these theoretical insights could help identify multipartite states that are potentially relevant for new applications of quantum information theory. In this sense, it should be stressed that, wherever the LOCC ordering is not trivial, this gives rise to a notion of maximal entanglement \cite{MESpaper}. Such states must then be the most useful states within the family for any protocol to be conducted in the general scenario where parties are restricted to performing LOCC manipulations independently of the particular task to be accomplished (e.g., teleportation, metrology, etc.).

One of the main results in this paper is that there is indeed another isolation-free class, namely the SLOCC class of the 3-qutrit totally antisymmetric state $\ket{A_3}$. Totally antisymmetric states are fermionic states of total spin zero \cite{AdanAntiSymmAppl}. They have applications in solving certain multiparty communication problems such as the Byzantine agreement \cite{ByzantineAgree}, $N$ strangers, secret sharing, and liar detection problems \cite{AdanAntiSymmAppl}. In more general settings, these states allow the inversion of unitary \cite{InvertUnitary} and isometry operations \cite{InvertIsometry} due to their symmetry properties. Given that $\ket{A_3}$ has many applications and has an isolation-free SLOCC class, it is then well motivated to take a closer look into the entanglement properties of the whole class. While studying allowed LOCC transformations within the class, we also show that the \textit{maximally entangled set} (MES) \cite{MESpaper} of this SLOCC class, which is the minimal set of LU-inequivalent states such that every state in the whole SLOCC class can be reached by a state from the set with LOCC, is of zero measure. In other words, the most ``useful'' set of states constitutes only a vanishingly small portion of the SLOCC class.

One of the main reasons that makes it possible to characterize LOCC convertibility between bipartite pure states in a simple way is that it turns out to be sufficient to consider very simple protocols. Prior to the work of Ref.\;\cite{Nielsen}, it was shown in Ref.\;\cite{LP} that in this case, the study of these transformations can be restricted to LOCC protocols with a single round of classical communication. Most of the known LOCC protocols for the multipartite case have a similarly simple structure. Although more rounds of classical communication might be necessary, these transformations are achieved by concatenating single-round transformations, i.e., the initial state is transformed to the final state by mapping it deterministically to intermediate pure states after each round of classical communication. However, it was shown in Ref.\;\cite{ProbStep1} (further examples can be found in Refs.\;\cite{ProbStep2,OurSymmPaper}) that the study of LOCC transformations in the multipartite case cannot be boiled down to protocols of this simple form. There exist LOCC conversions such that the input state cannot be transformed into the output state by concatenating deterministic one-round protocols and intermediate probabilistic steps are required. This adds another difficulty to the formidable problem of characterizing LOCC transformations within multipartite SLOCC classes with a non-trivial stabilizer. Nevertheless, the work in Refs.\;\cite{ProbStep1, ProbStep2} (see Ref.\;\cite{OurSymmPaper} as well) characterizes, in terms of the stabilizer, which states are one-round convertible or reachable by a protocol with a finite (but otherwise arbitrary) number of rounds of classical communication. This is the main technical ingredient that allows us to prove that the 3-qutrit totally antisymmetric SLOCC class is isolation-free.

Next, we consider whether this result can be extended to other SLOCC classes and we provide a plethora of no-go results. Notice that the aforementioned techniques have limitations in proving a given state to be isolated. With them, one can only aim at showing that a state is neither one-round convertible nor finite-round reachable, a property that is referred to as \textit{weak isolation}. Therefore, weak isolation does not necessarily imply isolation. However, to our knowledge, there is no evidence of the existence of a state that is weakly isolated but not isolated \cite{footnote:NOT1round}. We remark that if weak isolation exists in an isolation-free SLOCC class, then all the weakly isolated states would need to be non-isolated under highly contrived LOCC protocols such as those that converge to the desired target state in the limit where the number of rounds goes to infinity. Notice that the very existence of a particular LOCC transformation between pure multipartite states that can only be accomplished with this latter type of protocols remains an open question (see Ref.\;\cite{SEPeqn}). In any case, although it does not prove isolation, weak isolation is a strong indication of SLOCC classes with limited LOCC convertibility properties. In this part of the manuscript, we prove weak isolation for very general SLOCC classes including all classes that have a compact stabilizer and also the ones with a non-compact stabilizer such as all 4-qubit classes and the classes of $n$-qunit totally antisymmetric states for all $n\geq4$. This extends the results of Ref.\;\cite{OurSymmPaper}, which also found weak isolation in the SLOCC classes corresponding to permutation-symmetric states (with the sole exception of the GHZ and W classes). All of these suggest that even among the non-generic SLOCC classes that have a non-trivial stabilizer, the absence of isolation remains an elusive phenomenon.

The outline of this paper is as follows. In Sec.\;\ref{sec:preliminaries}, we first introduce some definitions and notations that we will use throughout the paper. Then, we recall some important results that can be used to determine convertibility under LOCC and whether a state is weakly isolated. Moreover, we review the reason why the $n$-qubit GHZ and W classes are isolation-free and some key properties of the totally antisymmetric states. In Sec.\;\ref{sec:results} where we present our main results, we prove that every SLOCC class that contains a fully entangled state with a finite or compact symmetry group contains weakly isolated states. The next thing we consider are SLOCC classes with non-compact infinite symmetry groups. In particular, we show that weak isolation exists in every SLOCC class that contains an $n$-qudit fully entangled state with a symmetry group $\{S^{\otimes n} \mbox{ with } S\in GL(d,\mathbb{C})\}$ for all $n\geq 4$. Examples for such classes would be the ones of $n$-qunit totally antisymmetric states with $d=n\geq 4$. However, remarkably, we prove that SLOCC classes that contain a $3$-qudit fully entangled state with a symmetry group $\{S^{\otimes 3} \mbox{ with } S\in SL(d,\mathbb{C})\}$ are isolation-free. The class of $\ket{A_3}$ is an example of such a class in the case $d=3$. Finally, we study the entanglement properties of the $\ket{A_3}$ class in Sec.\;\ref{sec:MES_A3}. We characterize the set of states in the $\ket{A_3}$ class that cannot be obtained from other LU-inequivalent states with $\text{LOCC}_\mathbb{N}$ ($M_{A_3}$) and rule out certain LOCC transformations within $M_{A_3}$. We also present a way to decompose and prepare all the states in $M_{A_3}$ with a qutrit, a 2-qubit entangled state embedded in a 2-qutrit Hilbert space, and two 2-qutrit entangling unitaries.

\section{Preliminaries}
\label{sec:preliminaries}

Before reviewing some key concepts in entanglement theory we introduce here our notation. Throughout this paper, the set $\{1,\ldots,n\}$ is denoted by $[n]$ for any $n\in\mathbb{N}$, $\sigma_x$ and $\sigma_z$ denote the Pauli $x$ and $z$ matrices, $GL(d,\mathbb{C})$ and $SL(d,\mathbb{C})$ denote the sets of $d\times d$ complex invertible and special linear matrices, respectively. As a shorthand notation, we denote a positive definite matrix $X$ by $X>0$.

In this work, we consider only $n$-qudit fully entangled pure states $\ket{\psi}\in\bigotimes_{j=1}^n\mathbb{C}^{d_j}$ where $d_j$ is the local dimension of the $j$-th qudit. All reduced density matrices of fully entangled states have full rank, i.e., $rk(\text{Tr}_{[n]\setminus\{j\}}(|\psi\rangle\langle\psi|))=d_j$ for all $j\in[n]$. Thus, we do not consider transformations to states supported on lower dimensions nor transformations where some subset of the parties is disentangled from the rest as this will clearly result in less entangled states (according to the Schmidt-rank measure \cite{SchmidtMeas}) which belong to entanglement classes inequivalent to the one containing the initial state under the framework \cite{GourWallach} that we use here. Note that Refs.~\cite{GourWallach,SEPeqn,OurSymmPaper} also focused on fully entangled states and their techniques, which we build our results on, require this hypothesis. The SLOCC class of a state $\ket{\psi}$ is defined to be the set $\{\bigotimes_{k=1}^n g_k\ket{\psi}:g_k\in GL(d_k,\mathbb{C})\;\forall \;k\in[n]\}$. Two states are SLOCC-equivalent if they belong to the same SLOCC class. The above restriction implies that we only study LOCC transformations within a given SLOCC class. In this case it is useful to choose a representative $\ket{\Psi_s}$ (also called a \textit{seed state}) from an SLOCC class and characterize its stabilizer. The \textit{stabilizer} (also called the \textit{symmetry group}) of a state $\ket{\psi}$ is the set of local operators $\mathcal{S}_{\psi}=\{S=\bigotimes_{k=1}^n S^{(k)}:S\ket{\psi}=\ket{\psi} \text{ with } S^{(k)}\in GL(d_k,\mathbb{C})\;\forall \;k\in[n]\}$. We say $\ket{\psi}$ is stabilized by $S$ if $S\in\mathcal{S}_{\psi}$. We also define the \textit{local symmetry group} of the $j$-th qudit of $\ket{\psi}$, $\mathcal{S}_{\psi}^{(j)}=\{S^{(j)}\in GL(d_j,\mathbb{C}):\exists\;S^{(i)}\;\forall\;i\in[n]\setminus\{j\} \text{ so that } \bigotimes_{k=1}^n S^{(k)}\in \mathcal{S}_{\psi}\}$ which can, for all but one site, be assumed to be a subset of $SL(d_j,\mathbb{C})$ \cite{footnote:LocalSymm_SL}.

Within the SLOCC class of $\ket{\Psi_s}$, we denote the initial and final states of a transformation as $\ket{\Psi}=g\ket{\Psi_s}$ and $\ket{\Phi}=h\ket{\Psi_s}$, respectively, where $g=\bigotimes_{k=1}^n g_k$ and $h=\bigotimes_{k=1}^n h_k$. We will predominantly work with unnormalized states unless specified otherwise. Since entanglement is invariant under LU transformations $\bigotimes_{i=1}^n U(d_i,\mathbb{C})$, the study of state transformations can be restricted to between LU-inequivalent states. Hence, it is convenient to consider the operators $G=\bigotimes_{i=1}^n G_i=g^\dagger g$ and $H=\bigotimes_{i=1}^n H_i = h^\dagger h$ which are invariant under the action of any (local) unitary on the respective states $g\ket{\Psi_s}$ and $h\ket{\Psi_s}$. Their LU-invariant property is the reason why they play an important role in characterizing non-LU state transformations as one can see from Eq.\;(\ref{eq:SEP}). It is also useful to define the set of local singular operators of a state $\ket{\Psi}$ as $\mathcal{N}_{\Psi}=\{N=\bigotimes_{k=1}^n N^{(k)}:N\ket{\Psi}=0\}$.

\subsection{Necessary condition for LOCC transformability}\label{sec:SEPconditions}

We will recall here some results regarding LOCC transfromations among pure multipartite states. Readers who are familiar with results regarding LOCC transformations can skip Sec.\;\ref{sec:SEPconditions}.

An LOCC transformation is described by a sequence of local completely positive trace-preserving maps whose outcomes can be postselected according to a particular Kraus decomposition of the map (i.e., a local measurement) and broadcast to the other parties by a classical communication channel. Any subsequent action in this sequence can be made dependent on the previously received classical information. Thus, based on this, each round specifies which local measurement is to be performed by which party and describes which LU operation is to be applied by the other parties, given the particular outcome of this round. The class of LOCC transformations that terminates after a given number $n\in\mathbb{N}$ of rounds will be denoted by LOCC$_\mathbb{N}$ (i.e., finite-round protocols). In particular, LOCC$_1$ will stand for the subclass of transformations in LOCC$_\mathbb{N}$ that consist of a single round. For more detailed definitions of LOCC, we refer the reader to Refs.\;\cite{SEP>LOCCRef1,DonaldHorodecki,LOCChardEric,LOCC_Everything,SEPeqn,OurSymmPaper}. The mathematical description and analysis of LOCC transformations is highly involved and one way to circumvent this difficulty is to study state transformations with an enlarged set of operations. A particularly convenient example is the set of  \textit{separable maps} (SEP), that strictly contains LOCC \cite{SEP>LOCCRef1,LOCC_Everything} and has a much simpler mathematical characterization despite not having any physical interpretation. A SEP is a completely positive trace-preserving map for which there exists a Kraus decomposition in which all the Kraus operators are local. The necessary and sufficient condition for the existence of a SEP transformation among pure states is stated in the following theorem which was proven in Refs.\;\cite{GourWallach,SEPeqn}.

\begin{theorem}[\cite{GourWallach,SEPeqn}]
\label{thm:SEP}
The state $g\ket{\Psi_s}$ can be transformed to $h\ket{\Psi_s}$ via SEP if and only if there exists a finite set of probabilities $\{p_k\}$, symmetries $\{S_k\} \subseteq \mathcal{S}_{\Psi_s}$, and $N_q\in\mathcal{N}_{g\Psi_s}$ such that
\begin{equation}
    \frac{1}{r}\sum_k p_k S_k^\dagger HS_k + g^\dagger\sum_q N_q^\dagger N_q g = G,\label{eq:SEP}
\end{equation}
where $r=||h\ket{\Psi_s}||^2/||g\ket{\Psi_s}||^2$.
\end{theorem}

From Theorem \ref{thm:SEP}, we see that the stabilizer plays a crucial role in determining whether two states from the same SLOCC class are related via SEP. If one can prove that one state cannot reach another state with SEP, then it also implies that the transformation is impossible with LOCC. For example, since generic states only have trivial symmetries, they are isolated under SEP \cite{Isolation1,Isolation2}, and hence isolated under LOCC transformations, which is the reason why we study non-trivial LOCC transformations only among zero-measure sets of states. Determining SEP transformability with Theorem \ref{thm:SEP} is often a non-trivial task since it is hard to fully characterize both the symmetry group $\mathcal{S}_{\Psi_s}$ and the set of operators $\mathcal{N}_{g\Psi_s}$, which annihilate the state, in general. Moreover, even if a SEP transformation is proven to exist between two states, it does not imply that an LOCC protocol exists. Therefore, Theorem \ref{thm:SEP} serves only as a necessary condition for the existence of LOCC transformations.

\subsection{$\text{LOCC}_\mathbb{N}$-reachability and weak isolation}\label{sec:LOCCN_reachability}

A state is isolated if it is not \textit{reachable} nor \textit{convertible} via LOCC (excluding LU transformations). That is, it cannot be obtained by transforming any SLOCC-equivalent state, nor can it be transformed into any other SLOCC-equivalent state with LOCC. As pointed out in Sec.\;\ref{sec:SEPconditions}, only a necessary condition for multipartite LOCC transformability is known, so there is no clear way of determining whether a state is isolated except for generic states that only have trivial symmetries \cite{Isolation1,Isolation2}. However, one could prove a weaker form of isolation by considering a restricted set of transformations--$\text{LOCC}_{\mathbb{N}}$, which are, from a practical point of view, more relevant and characterized better. The necessary and sufficient condition for a state to be reachable from an SLOCC-equivalent state via $\text{LOCC}_{\mathbb{N}}$ is stated in Theorem \ref{thm:reachability}. Its proof in Ref.\;\cite{OurSymmPaper} relies on the fact that there exists a final round in any $\text{LOCC}_{\mathbb{N}}$ protocol.

\begin{theorem}[\cite{OurSymmPaper}]
\label{thm:reachability}
A state $\ket{\Phi}\propto h \ket{\Psi_s}$ is reachable via $\text{LOCC}_{\mathbb{N}}$, iff there exists $S \in \mathcal{S}_{\Psi_s}$ such that the following conditions hold up to permutations of the particles:
\begin{itemize}
 \item[(i)] For any $i\geq 2$, $ (S^{(i)})^\dagger H_i S^{(i)} \propto H_i$ and
\item[(ii)] $  (S^{(1)})^\dagger H_1 S^{(1)} \not\propto H_1$.
\end{itemize}
\end{theorem}

A simple example for an $\text{LOCC}_{\mathbb{N}}$-reachable state would be a 3-qubit state $\ket{\psi}\propto h_1\otimes\identity^{\otimes2}\ket{\Psi_s}$, where $h_1=\frac{1}{2}\identity+\sigma_x$ and $\ket{\Psi_s}$ is any 3-qubit state with a symmetry $\sigma_z^{\otimes3}$, as $\sigma_z$ cannot \textit{quasi-commute} \cite{footnote:QuasiComm} with $H_1=h_1^\dagger h_1=(\frac{1}{2}\identity+\sigma_x)^2=\frac{5}{4}\identity+\sigma_x$ (i.e., $\sigma_z^\dagger H_1\sigma_z\not\propto H_1$). In this case, $\ket{\psi}$ can be reached from $\ket{\Psi_s}$ with a one-round LOCC protocol where the first party measures with the operators $\{\sqrt{\frac{2}{5}}h_1\sigma_z^m\}_{m=0,1}$ and parties 2 and 3 each apply LU $\sigma_z^m$ depending on the measurement outcome $m$. For more examples of states that are (not) $\text{LOCC}_{\mathbb{N}}$-reachable, see e.g., Ref.\;\cite{OurSymmPaper} and Secs.\;\ref{sec:results} and \ref{sec:MES_A3}.

However, to determine whether a state is $\text{LOCC}_{\mathbb{N}}$-convertible is generally much harder due to, among other reasons, the existence of protocols that require intermediate probabilistic steps \cite{ProbStep1,ProbStep2,OurSymmPaper}. Instead, if one asks only for LOCC$_1$ convertibility, then one can use the necessary and sufficient conditions given by Lemma 1 in Ref.\;\cite{OurSymmPaper}. We can define a weaker form of isolation by calling a state \textit{weakly isolated} if it is not $\text{LOCC}_{\mathbb{N}}$-reachable and not $\text{LOCC}_1$-convertible. The necessary and sufficient condition for weak isolation is given by the following lemma~\cite{OurSymmPaper}.

\begin{lemma}[\cite{OurSymmPaper}]
\label{lemma:isolation}
A state $\ket{\Phi}\propto g \ket{\Psi_s}$ is weakly isolated if and only if there exists no $S \in \mathcal{S}_{\Psi_s} \setminus\{\identity\}$ such that $S^{(i)}$ \textit{quasi-commutes} with $G_i$ [i.e., $ (S^{(i)})^\dagger G_i S^{(i)} \propto G_i$] for at least $n-1$ sites $i$.
\end{lemma}

To give an example, a 3-qubit state $\ket{\psi}$ with a symmetry $\sigma_z^{\otimes 3}$, which (quasi-)commutes with $G_i=\identity$ for all 3 sites, is not (weakly) isolated. It can be transformed into another state $h_1\otimes\identity^{\otimes2}\ket{\psi}$ with $h_1\in GL(2,\mathbb{C})$ such that $H_1$ is not diagonal and $\text{diag}(H_1)=\identity$. The transformation is achieved by having party 1 measure the first qubit of $\ket{\psi}$ with operators $\{\frac{1}{\sqrt{2}}h_1\sigma_z^m\}_{m=0,1}$ and parties 2 and 3 apply a local unitary $\sigma_z^m$ based on the measurement outcome $m$. More examples of (non-)weakly isolated states can be found in Ref.\;\cite{OurSymmPaper} and in Secs.\;\ref{sec:results} and \ref{sec:MES_A3}.

With Lemma \ref{lemma:isolation}, one has a relatively simple mathematical criterion to show that an SLOCC class is isolation-free by proving that every state it contains is not weakly isolated since any isolated state is also weakly isolated by definition. In fact, this is exactly how we identify a new isolation-free class, the SLOCC class of the 3-qutrit totally antisymmetric state (Sec.\;\ref{sec:noncompact}). As already mentioned in the introduction, we should note that weak isolation is only an indicator but not an implication for full LOCC isolation.

\subsection{Known isolation-free SLOCC classes}\label{sec:KnownIsolationFreeClasses}

Before starting the search for new isolation-free SLOCC classes, let us revisit the key properties of the $n$-qubit GHZ and W states, $\ket{\text{GHZ}_n}$ and $\ket{\text{W}_n}$, whose SLOCC classes do not contain any isolated states, i.e., the SLOCC classes are isolation-free. The symmetry groups of the two states are given by \cite{WGHZsymm}
\begin{align}
    &\mathcal{S}_{\text{GHZ}} = \left\{\bigotimes_{i=1}^n P_{z_i}\sigma_x^m: m\in\{0,1\}, z_i\in\mathbb{C}\setminus\{0\}, \prod_{i=1}^n z_i = 1\right\},\label{eq:S_GHZ}\\
    &\mathcal{S}_{\text{W}} = \left\{\frac{1}{x}\bigotimes_{i=1}^n T_{x,y_i}: x\in\mathbb{C}\setminus\{0\}, y_i\in\mathbb{C}, \sum_{i=1}^n y_i = 0\right\}\label{eq:S_W}
\end{align}
where $P_{z_i}=diag(z_i, z_i^{-1})$ and $T_{x,y_i}=\begin{pmatrix} 1 & y_i \\ 0 & x \end{pmatrix}$. One can verify that these two classes do not have any weakly isolated states using the following observations. Let $H_i=\begin{pmatrix} a_i & b_i \\ b_i^* & c_i \end{pmatrix}>0$. It holds that $\sigma_xP_{z_i}^\dagger H_i P_{z_i}\sigma_x\propto H_i$ if and only if $z_i=\pm\sqrt{\frac{c_i b_i}{a_i b_i^*}}$ and $T_{x,y_i}^\dagger H_i T_{x,y_i}\propto H_i$ if and only if $y_i=\frac{b_i}{a_i}(1-x)$ and $|x|=1$. Since $n-1$ parameters $z_i\in\mathbb{C}\setminus\{0\}$ ($y_i\in\mathbb{C}$) are independent in $\mathcal{S}_{\text{GHZ}}$ ($\mathcal{S}_{\text{W}}$), for any $H=\otimes_{i=1}^n H_i$, there always exists a non-trivial symmetry from $\mathcal{S}_{\text{GHZ}}$ ($\mathcal{S}_{\text{W}}$) that quasi-commutes with $H$ for $n-1$ sites, respectively. Therefore, due to Lemma \ref{lemma:isolation}, the two SLOCC classes do not have any state that is (weakly) isolated.

\subsection{Totally antisymmetric states}\label{sec:An_Properties}
We now introduce some properties of totally antisymmetric states, which we will need when we prove (non-)existence of weakly isolated states in their SLOCC classes in Sec.\;\ref{sec:noncompact} and when we investigate entanglement properties of the SLOCC class of the 3-qutrit totally antisymmetric state in Sec.\;\ref{sec:MES_A3}.

The $n$-qunit totally antisymmetric state is given by
\begin{align}
    \ket{A_n} &=\frac{1}{\sqrt{n!}}\sum_\pi sgn(\pi)P_\pi(\ket{1,2,\ldots,n}) \label{eq:An_sign}\\
    &= \frac{1}{\sqrt{n!}}\sum_{i_1,\ldots,i_n=1}^n \epsilon_{i_1,\ldots,i_n} |i_1,\ldots,i_n\rangle\label{eq:An_epsilon}
\end{align}
where the sum in Eq.\;(\ref{eq:An_sign}) is over all permutations of $n$ elements $\pi\in S_n$, $sgn(\pi)$ is the sign of $\pi$, $P_\pi(\ket{1,2,\ldots,n})=\ket{\pi(1),\pi(2),\ldots,\pi(n)}$ and $\epsilon_{i_1,\ldots,i_n}$ is the Levi-Civita symbol. The key property of $\ket{A_n}$ that we will use later is its invariance under permutation up to a sign. More precisely, it is easy to see that 
\begin{equation}
    P_{\tilde{\pi}}\ket{A_n} = sgn(\tilde{\pi})\ket{A_n}\label{eq:An_perm}
\end{equation}
using the properties of $sgn$: $sgn(\pi)=sgn(\tilde{\pi}^{-1}\tilde{\pi}\pi)=sgn(\tilde{\pi}^{-1})sgn(\tilde{\pi}\pi)$ and $sgn(\tilde{\pi}^{-1})=sgn(\tilde{\pi})$.

As we saw in Secs.\;\ref{sec:SEPconditions}--\ref{sec:KnownIsolationFreeClasses}, in order to determine LOCC transformability and whether weak isolation exists in an SLOCC class, it is essential to have a full characterization of the stabilizer of a seed state from that class. Hence, we need to characterize the symmetry group of $\ket{A_n}$, which is stated in the following observation. Its proof can be found in Appendix \ref{app:An_symm}.

\begin{observation}\label{obs:An_symm_group}
The symmetry group of the $n$-qunit totally antisymmetric state $\ket{A_n}$ is $\{S^{\otimes n} \mbox{ with } S\in SL(n,\mathbb{C})\}$.
\end{observation}

\section{Identifying isolation-free SLOCC classes}
\label{sec:results}

In this section, we will first focus on SLOCC classes that contain a fully entangled state with a compact symmetry group and prove that all these classes contain weakly isolated states. Then, we move on to consider classes with non-compact stabilizers and prove that for any local dimension $d$, all 3-qudit SLOCC classes with a state stabilized by $S^{\otimes3}$ for all $S\in SL(d,\mathbb{C})$ are isolation-free. One such example is the SLOCC class of the 3-qutrit totally antisymmetric state $\ket{A_3}$. Finally, we show that for $n\geq4$, if a state in an $n$-qudit SLOCC class has all symmetries of the form $S^{\otimes n}$, then the class contains weak isolation. Examples of these classes are the ones containing $n$-qunit totally antisymmetric states $\ket{A_n}$. In a slight abuse of notation, here and in the rest of the paper, we use $\bigotimes_{i=1}^n U(d_i,\mathbb{C})$ to denote the group with elements $\bigotimes_{i=1}^n U_i$ where each $U_i$ is an arbitrary unitary matrix in $U(d_i,\mathbb{C})$ for all $i$ (and similarly for $\bigotimes_{i=1}^n GL(d_i,\mathbb{C})$). 

\subsection{SLOCC classes with compact stabilizers}\label{sec:compact}

We consider here SLOCC classes whose representatives possess a compact symmetry group. We call a set of $d\times d$ complex matrices \textit{closed} (\textit{compact}) if it is closed (compact) in the Euclidean topology on $\mathbb{C}^{d^2}$ unless specified otherwise. As stated in the following theorem, we prove that all SLOCC classes containing a fully entangled state which is stabilized by a compact stabilizer contains weakly isolated states.

\begin{theorem}\label{thm:WI_CompactSymm}
Every $n$-qudit SLOCC class that contains a fully entangled state $\ket{\psi}\in\bigotimes_{i=1}^n\mathbb{C}^{d_i}$ which has a compact symmetry group $\mathcal{S}_\psi \subset \bigotimes_{i=1}^n GL(d_i,\mathbb{C})$ contains weakly isolated states.
\end{theorem}

In order to prove this, we need the following proposition, which states that a compact symmetry group is always finite, and it is proven in Ref.\;\cite{GourWallach}.

\begin{proposition}[\cite{GourWallach}]
\label{prop:Compact-Finite}
Let $\ket{\psi}$ be a state with a symmetry group $\mathcal{S}_\psi$ that is compact, then $\mathcal{S}_\psi$ is finite.
\end{proposition}

\begin{proof}[Proof of Theorem \ref{thm:WI_CompactSymm}]
Due to Proposition \ref{prop:Compact-Finite}, we are left to prove that any SLOCC class that contains a fully entangled state $\ket{\psi}$ with a finite symmetry group $\mathcal{S}_{\psi}$ contains weakly isolated states. Let us first recall that whenever there exists a state with finitely many symmetries, then there exists another state $\ket{\phi}$ in the same SLOCC class with finitely many unitary symmetries \cite{GourWallach}. Hence, it remains to show that if a state has only finitely many unitary symmetries, then there exists weak isolation. In order to show this, we denote a non-empty finite set of non-trivial $d\times d$ unitary matrices (i.e., $U\not\propto\identity$) by $\mathcal{F}$ and show that there exists $P>0$ such that $U^\dagger P U \not\propto P$ for all $U\in\mathcal{F}$. This can be easily seen from the fact that all $U\in\mathcal{F}$ have a finite set of eigenspaces. Hence, choosing $P$ as a positive operator with an eigenvector that is not in any of these finitely many eigenspaces proves the statement \cite{footnote:PosOpNotCommUs}. Since $\mathcal{S}_\phi$ is a finite unitary group, the corresponding local symmetry group $\mathcal{S}_\phi^{(i)}$ for every site $i\in[n]$ is also a finite unitary group. We can construct a state $\bigotimes_i \sqrt{P_i}\ket{\phi}$ such that $U^\dagger P_i U\not\propto P_i$ for all non-trivial $U\in\mathcal{S}_\phi^{(i)}$ and will show that it is weakly isolated. Given that the state is not weakly isolated if and only if there exists a non-trivial symmetry that (quasi-)commutes with $P_i$ for all but one site, it remains to show that no fully entangled state can have a symmetry of the form $S^{(j)}\otimes_{k\neq j} \identity_{d_k}\in\mathcal{S}_{\phi}$ unless it is trivial, i.e., $S^{(j)}=\identity_{d_j}$. This last statement follows easily by considering the Schmidt decomposition in the splitting, system $j$ versus the rest. Therefore, $\bigotimes_i \sqrt{P_i}\ket{\phi}$ is weakly isolated.\end{proof}

Since a unitary symmetry group is also compact (see also Appendix \ref{app:proof2ThmU_Symm}), the following corollary follows immediately from Theorem \ref{thm:WI_CompactSymm}.

\begin{corollary}\label{cor:UnitarySymm_WI}
Every $n$-qudit SLOCC class that contains a fully entangled state $\ket{\psi}$ which has a unitary symmetry group $\mathcal{S}_\psi \subseteq \bigotimes_{i=1}^n U(d_i,\mathbb{C})$ contains weakly isolated states.
\end{corollary}

The proof of Proposition \ref{prop:Compact-Finite} in Ref.\;\cite{GourWallach} requires some familiarity with algebraic geometry. However, Corollary \ref{cor:UnitarySymm_WI} can also be proven for $d_j=2\;\forall\;j\in[n]$ without relying on Proposition \ref{prop:Compact-Finite}. In Appendix \ref{app:proof2ThmU_Symm}, we provide an alternative proof which uses more elementary mathematics and contains several observations about the general structure of stabilizers that might be of independent interest to some readers.

To summarize all the ideas behind the proof, weak isolation exists in all SLOCC classes with compact stabilizers because there are not enough symmetries in any compact stabilizer to quasi-commute with every possible $\bigotimes_{i=1}^n H_i>0$ for $n-1$ qudits and by Lemma \ref{lemma:isolation}, this leads to weak isolation. Although weak isolation does not imply isolation, these results suggest that it is more promising to investigate SLOCC classes with states stabilized by a non-compact symmetry group to identify isolation-free SLOCC classes.

\subsection{SLOCC classes with non-compact stabilizers}\label{sec:noncompact}

We have seen that non-compact stabilizers are necessary to identify SLOCC class without weak isolation. In this section we show that tripartite states with a particular non-compact stabilizer are indeed isolation-free. An example of such an SLOCC class is given by the totally antisymmetric state. Extending such stabilizers to more parties leads, however, again to weak isolation, as we show at the end of this section. 

\begin{theorem}\label{thm:3Qudit_IsoFree}
Let $\ket{\Psi_s}$ be a $3$-qudit state with a symmetry group $\mathcal{S}_{\Psi_s}=\{S^{\otimes 3} \mbox{ with } S\in SL(d,\mathbb{C})\}$, with $d\geq 2$. Then, no weakly isolated state exists within the SLOCC class of $\ket{\Psi_s}$. In particular, the SLOCC class of the 3-qutrit totally antisymmetric state $\ket{A_3}$ is isolation-free.
\end{theorem}
\begin{proof}
Using that $\frac{\tilde{g}_3^{\otimes 3}}{(\det \tilde{g}_3)^{3/d}}\in \mathcal{S}_{\Psi_s}$, for any $\tilde{g}_3$, we have that any state in the SLOCC class of $\ket{\Psi_s}$ can be written as 
\begin{equation}
    \tilde{g}_1\otimes \tilde{g}_2\otimes \tilde{g}_3\ket{\Psi_s} \propto g_1 \otimes g_2 \otimes \mathbbm{1} \ket{\Psi_s} \label{eq:A3_SLOCC}
\end{equation}
where $\tilde{g}_i\in GL(d,\mathbb{C})$ and $g_i=\tilde{g}_i\tilde{g}_3^{-1}$ for $i=1,2$. The proportionality factor in the equation above comes from $\tilde{g}_3^{\otimes 3} \ket{\Psi_s}=(\det \tilde{g}_3)^{3/d}\ket{\Psi_s}$. To show that any state of the form of Eq.\;(\ref{eq:A3_SLOCC}) is not weakly isolated, we define $G_i=g_i^\dagger g_i$ for $i=1,2$ and use Lemma \ref{lemma:isolation}. Since the symmetry group of $\ket{\Psi_s}$ is $\{S^{\otimes 3} \mbox{ with } S\in SL(d,\mathbb{C})\}$, there always exists a unitary $S\not\propto\mathbbm{1}$ that satisfies $S^\dagger \mathbbm{1} S \propto \mathbbm{1}$ and $S^\dagger G_i S \propto G_i$ for $i=1$ or 2 by choosing the eigenbasis of $S$ identical to that of the positive operator $G_i$. As the quasi-commutation relation holds for at least 2 out of 3 sites, any state in form of Eq.\;(\ref{eq:A3_SLOCC}) is not weakly isolated by Lemma \ref{lemma:isolation}. Since the 3-qutrit totally antisymmetric state $\ket{A_3}$ has the symmetry group $\{S^{\otimes 3} \mbox{ with } S\in SL(3,\mathbb{C})\}$ by Observation \ref{obs:An_symm_group}, its SLOCC class has no (weakly) isolated states.
\end{proof}

This makes the SLOCC class of $\ket{A_3}$ the third known isolation-free class that has identical dimension in all local systems. Note however, that this property does not extend to larger numbers of parties, as we show in the following theorem, which is proven in Appendix \ref{app:proof_Sn_NOweakIsolation}.

\begin{theorem}\label{thm:Sn_weakIsolation}
Let $\ket{\Psi_s}$ be an $n$-qudit state $(n\geq4)$ with a symmetry group $\mathcal{S}_{\Psi_s}\subseteq\{S^{\otimes n} \mbox{ with } S\in GL(d,\mathbb{C})\}$ with $d\geq 2$ arbitrary. Then, there exist weakly isolated states within the SLOCC class of $\ket{\Psi_s}$.
\end{theorem}

Theorem \ref{thm:Sn_weakIsolation} and the fact that the symmetries of the $n$-qunit totally antisymmetric states are of the form $S^{\otimes n}$ implies that the SLOCC classes of the $n$-qunit totally antisymmetric states, $\ket{A_n}$, with $n\geq4$ contain weakly isolated states.

Hence, in order to find new isolation-free SLOCC classes, one needs to consider states with non-compact stabilizer different than $S^{\otimes n}$. However, a more general symmetry group does not preclude isolation. In fact, we show next that all SLOCC classes of $n$-qubit states with $n\leq4$, with the exception of the GHZ and W classes, contain weakly isolated states. 

\begin{observation}
The only isolation-free SLOCC classes for $n\leq4$ qubits that contain fully entangled states are the GHZ and W classes.\label{obs:nleq4_WI}
\end{observation}
\begin{proof}
Considering only fully entangled state, we have for 2 qubits, only the SLOCC class of the Bell state $\ket{\Phi^+} (=\ket{\text{GHZ}_2})$, and for 3 qubits, the ones of $\ket{\text{GHZ}_3}$ and $\ket{\text{W}_3}$ \cite{3qubitSLOCC}. For 4 qubits, despite that there are infinitely many SLOCC classes with fully entangled states \cite{4qubitSLOCC}, the stabilizer for every class was characterized in Ref.\;\cite{4qubitMES}. By Theorem \ref{thm:WI_CompactSymm} and Proposition \ref{prop:Compact-Finite}, we only need to show that weak isolation exists in classes with (non-compact) infinite stabilizers. In Appendix \ref{app:4qubit_WI}, we show that for every infinite stabilizer that corresponds to a 4-qubit fully entangled state (except for states in the W or GHZ classes) as characterized in Ref.\;\cite{4qubitMES}, one can always find $\otimes_{i=1}^4 H_i>0$ with which no non-trivial symmetry can quasi-commute for more than 2 sites.
\end{proof}

Notice that such a case-by-case analysis cannot be immediately extended to more parties. A full characterization of all $n$-qubit SLOCC classes and their stabilizers for $n>4$ is a formidable problem that, to the best of our knowledge, remains open so far.

\section{Entanglement properties of the $\ket{A_3}$ SLOCC class}\label{sec:MES_A3}

In this section, we investigate the entanglement properties of the SLOCC class of the 3-qutrit totally antisymmetric state $\ket{A_3}$. In particular, we will identify a superset, $M_{A_3}$, of the MES of this SLOCC class. It comprises of states that are not $\text{LOCC}_\mathbb{N}$-reachable and can collectively reach any SLOCC-equivalent state with LOCC in at most 3 rounds. Then, we will show that within the $M_{A_3}$, LOCC transformations from a full-measure subset to a zero-measure subset are impossible even with infinite-round protocols (Lemma \ref{lemma:MES_LOCCN_SEPforbid}). Finally, we show that all states in the $M_{A_3}$ can be obtained (up to LU) from applying two 2-qutrit entangling unitaries on a qutrit and a 2-qubit entangled state embedded in the Hilbert space of 2 qutrits.

\begin{observation}\label{obs:DiagStatesReachAll}
Any state in the SLOCC class of $\ket{A_3}$ is either LU equivalent to or can be reached from a state in one of the following three forms (up to LU) via LOCC in at most 3 rounds:
\begin{enumerate}[(i)]
    \item $\ket{A_3}$,
    \item $\ket{\psi_2(\alpha_1;\alpha_2)}\propto\text{diag}(\sqrt{\alpha_1},1,1)\otimes\text{diag}(\sqrt{\alpha_2},1,1)\otimes\mathbbm{1}\ket{A_3}$ where $\alpha_1,\alpha_2>0$, $\alpha_1\neq\alpha_2$ and $\alpha_1,\alpha_2\neq1$,
    \item $\ket{\psi_3(\alpha_1,\beta_1;\alpha_2,\beta_2)}\propto\text{diag}(\sqrt{\alpha_1},\sqrt{\beta_1},1)\otimes\text{diag}(\sqrt{\alpha_2},\sqrt{\beta_2},1) \otimes \mathbbm{1}\ket{A_3}$ where $\alpha_1,\alpha_2,\beta_1,\beta_2>0$, $\alpha_1\neq\beta_1\neq\beta_2\neq\alpha_2$, $\alpha_1\neq\alpha_2$, $\frac{\alpha_1}{\beta_1}\neq\frac{\alpha_2}{\beta_2}$ and $\alpha_1,\alpha_2,\beta_1,\beta_2\neq1$.
\end{enumerate}
Moreover, these states are not reachable via $\text{LOCC}_\mathbb{N}$. Hence, the MES of $\ket{A_3}$ is contained in the union set of all these states, which we call $M_{A_3}$.
\end{observation}

We prove this observation in Appendix \ref{app:DiagStatesReachAll}. Note that within the SLOCC class of $\ket{A_3}$, many states can be reached via LOCC. In fact, the MES of the $\ket{A_3}$ class is a zero-measure subset of the whole SLOCC class since the MES is contained in $M_{A_3}$ which is parametrized by 4 positive real numbers (see Observation \ref{obs:DiagStatesReachAll}), while general states in the SLOCC class require 8 real parameters \cite{footnote:MA3_zero_measure}.

Let us note that some transformations from states in $M_{A_3}$ to other states in the SLOCC class are achievable with non-trivial 3-round LOCC protocols. We illustrate that with an example: the transformation from $\ket{A_3}$ to $\ket{\psi}\propto h_1\otimes \sqrt{D_2}\otimes\identity\ket{A_3}$ where $h_1=\frac{1}{\sqrt{6}}\begin{pmatrix}
\sqrt{27+\sqrt{3}} & \sqrt{3-\sqrt{3}}\\
\sqrt{3-\sqrt{3}} & \sqrt{3+\sqrt{3}}
\end{pmatrix}\oplus1$ and $D_2=\text{diag}(25,9,1)$. First, party 3 applies a measurement to the third qutrit of $\ket{A_3}$ with operators $\{\sqrt{\frac{5}{11}}\text{diag}(\frac{1}{\sqrt{5}},1,1)X^j\}_{j=0,1,2}$ where $X=\sum_{i=0}^2|(i+1)\text{mod}\;3\rangle\langle i|$ and parties 1 and 2 each apply $X^j$ depending on the measurement outcome $j$ to get $\mathbbm{1}\otimes\mathbbm{1}\otimes\text{diag}(\frac{1}{\sqrt{5}},1,1)\ket{A_3} \propto \text{diag}(\sqrt{5},1,1)^{\otimes2}\otimes\mathbbm{1}\ket{A_3}$. Next, party 2 measures the state with operators $\{\frac{1}{\sqrt{10}}\text{diag}(5,3,1)\tilde{U}^k \text{diag}(\frac{1}{\sqrt{5}},1,1)\}_{k=0,1}$ where $\tilde{U}=1\oplus\begin{pmatrix}
0& -1 \\
1 & 0 
\end{pmatrix}$ and parties 1 and 3 each apply $\tilde{U}^k$ depending on the measurement outcome $k$. Finally, party 1 measures with $\{\frac{1}{\sqrt{3}}h_1 \text{diag}(\frac{1}{\sqrt{5}},1,1)Z^r\}_{r=0,1,2}$ and each of the remaining parties apply $Z^r$ depending on the measurement outcome $r$, where $Z=\text{diag}(1,e^{i\frac{2\pi}{3}},e^{i\frac{4\pi}{3}})$. Despite the fact that there are 
certain instances [e.g., $h_1$ and $D_2$ such that $\text{diag}(h_1^\dagger h_1)=D_2$] where a 2-round protocol is sufficient, it is worth noting that in this protocol, all three parties act non-trivially despite $g_3=h_3$. Other instances for state transformations that can be achieved with LOCC protocols using more than 2 rounds of communication in a non-trivial way are presented in  Ref.\;\cite{NonTrivial4round}.

In the process of characterizing the MES of the $\ket{A_3}$ class, we show that for certain subsets of $M_{A_3}$, some states are not related even via SEP (and, hence, by LOCC). More specifically, we prove in Lemma \ref{lemma:MES_LOCCN_SEPforbid} that none of the type-(iii) states in $M_{A_3}$ (see Observation \ref{obs:DiagStatesReachAll}) can be converted to any type-(ii) states with SEP. In addition, we show that certain type-(ii) states are not interconvertible via SEP. For more details, please refer to Lemma \ref{lemma:MES_LOCCN_SEPforbid} in Appendix \ref{app:SEPforbid}. This result together with Observation \ref{obs:DiagStatesReachAll} are summarized in Fig.\;\ref{Fig:MES_A3}.

\begin{figure}[h!]
    \centering
    \includegraphics[width=\linewidth]{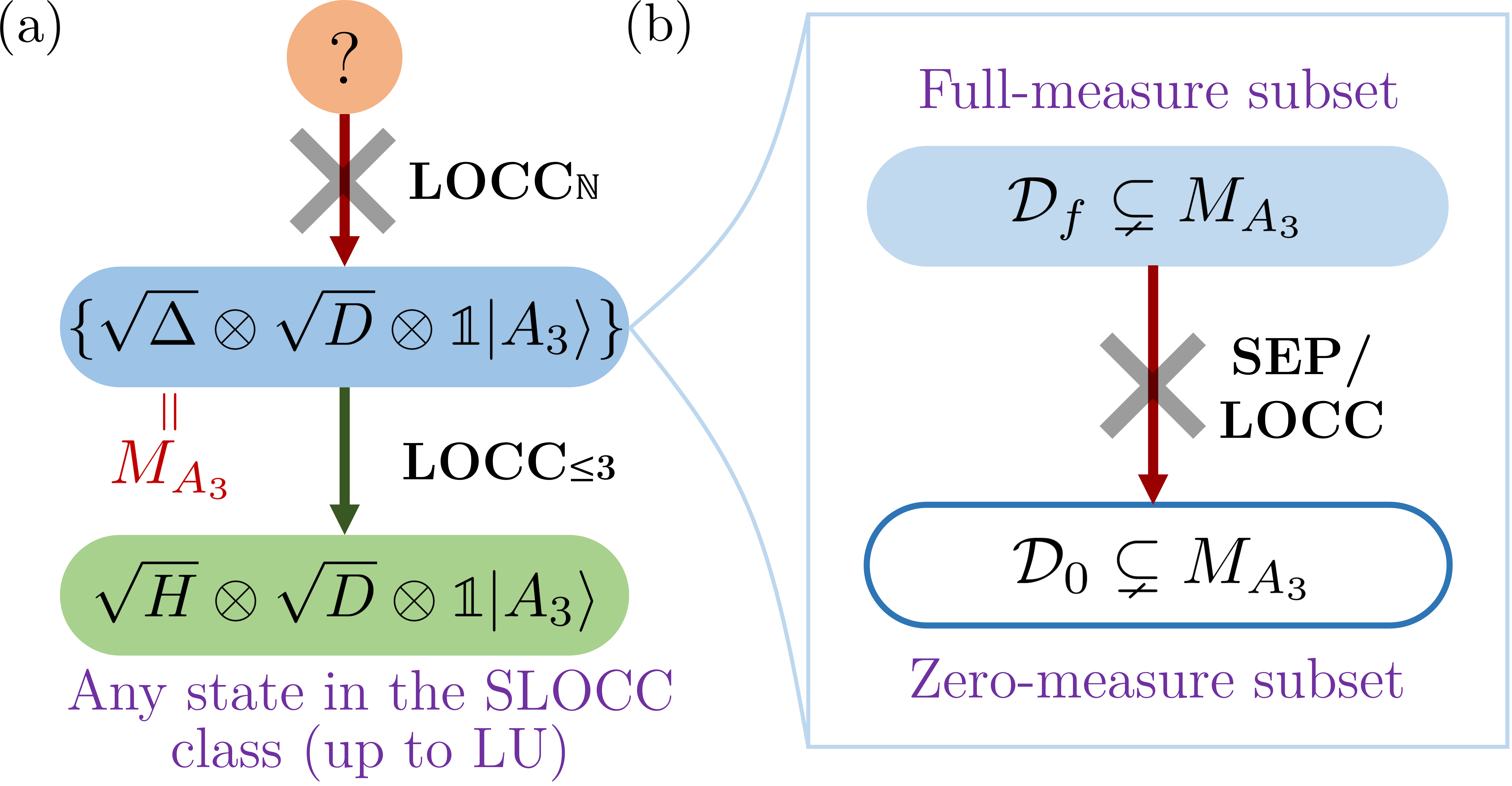}
    \caption{The figure summarizes the results about state transformations within the SLOCC class of $\ket{A_3}$. (a) Any state in the SLOCC class can be reached (up to LU) from a state in $M_{A_3}$ with LOCC using at most 3 rounds of communications, whereas all states in $M_{A_3}$ are unreachable from any other states with $\text{LOCC}_{\mathbb{N}}$. (b) Within $M_{A_3}$, there exist a full-measure subset of states, $\mathcal{D}_f$, that cannot reach states in another zero-measure subset, $\mathcal{D}_0$, with SEP and hence, also not with LOCC. The definitions of sets $\mathcal{D}_f$ and $\mathcal{D}_0$ are given in Appendix \ref{app:SEPforbid}.}\label{Fig:MES_A3}
\end{figure}

The proof of these results relies on the satisfiability of a necessary but not sufficient condition for SEP transformations. For the most general transformations within $M_{A_3}$, arbitrary operators $g$ and $h$ corresponding to states therein and the set $\mathcal{N}_{g\ket{A_3}}$ need to be considered. As shown in Appendix \ref{app:SEPforbid}, for the state transformations presented here, a simple necessary condition for the existence of a SEP transformation, which is sufficient to show that these transformations do not exist, can be derived without specifying this set. 
Hence, we leave the task of fully characterizing the MES of the $\ket{A_3}$ class as an open problem.

Given the simple form of states in $M_{A_3}$, one may expect that these states can be completely characterized by bipartite entanglement across two different splittings of the three particles. Indeed, we find that all states in $M_{A_3}$ can be generated (up to LU) by entangling a qutrit with a 2-qubit entangled state embedded in the Hilbert space of 2 qutrits using two 2-qutrit controlled-unitaries following a procedure similar to the one proposed by Ref.\;\cite{Decomp3qubits} (see Ref.\;\cite{SubradianEntAtoms} for a different preparation scheme of $\ket{A_3}$). The preparation procedure of a state that is LU equivalent to the state $\ket{\psi(\alpha_1,\alpha_2,\beta_1,\beta_2)}\propto\text{diag}(\sqrt{\alpha_1},\sqrt{\alpha_2},1)\otimes\text{diag}(\sqrt{\beta_1},\sqrt{\beta_2},1) \otimes \mathbbm{1}\ket{A_3}$ \cite{footnote:ReorderAlpha2Beta1} is illustrated in Fig.\;\ref{Fig:Decomp_MES_A3} and is outlined as follows. First, qutrit 1 is prepared in the state $\ket{\widetilde{+}}=\frac{1}{\sqrt{3}}(\ket{0}+\ket{1}+\ket{2})$ and qutrits 2 and 3 are prepared in an entangled state $\ket{\psi_s}=\sqrt{\lambda_+}\ket{00}+\sqrt{\lambda_-}\ket{11}$ where $\lambda_\pm$ depends on $\alpha_1,\alpha_2,\beta_1,\beta_2$. Then, apply two commuting unitaries $U_{1j}=\sum_{k=0}^2|k\rangle\langle k|\otimes (U_j^\dagger D_\omega U_j)^k$ on qutrits 1 and $j$ for $j=2,3$ 
where $D_\omega=\text{diag}(1,e^{-i\frac{2\pi}{3}},e^{i\frac{2\pi}{3}})$. The exact forms of $\lambda_\pm$, $U_2$ and $U_3$ are stated in Appendix \ref{app:decomp_MES_A3}. With these parametrization,
\begin{equation}
    U_{13}U_{12}\ket{\widetilde{+}}_1\ket{\psi_s}_{23} \stackrel{\text{LU}}{\simeq}\ket{\psi(\alpha_1,\alpha_2,\beta_1,\beta_2)}
\end{equation} 
for any $\alpha_1,\alpha_2,\beta_1,\beta_2>0$, where the symbol $\stackrel{\text{LU}}{\simeq}$ stands for LU equivalence.
\begin{figure}[h!]
    \centering
    \includegraphics[width=\linewidth]{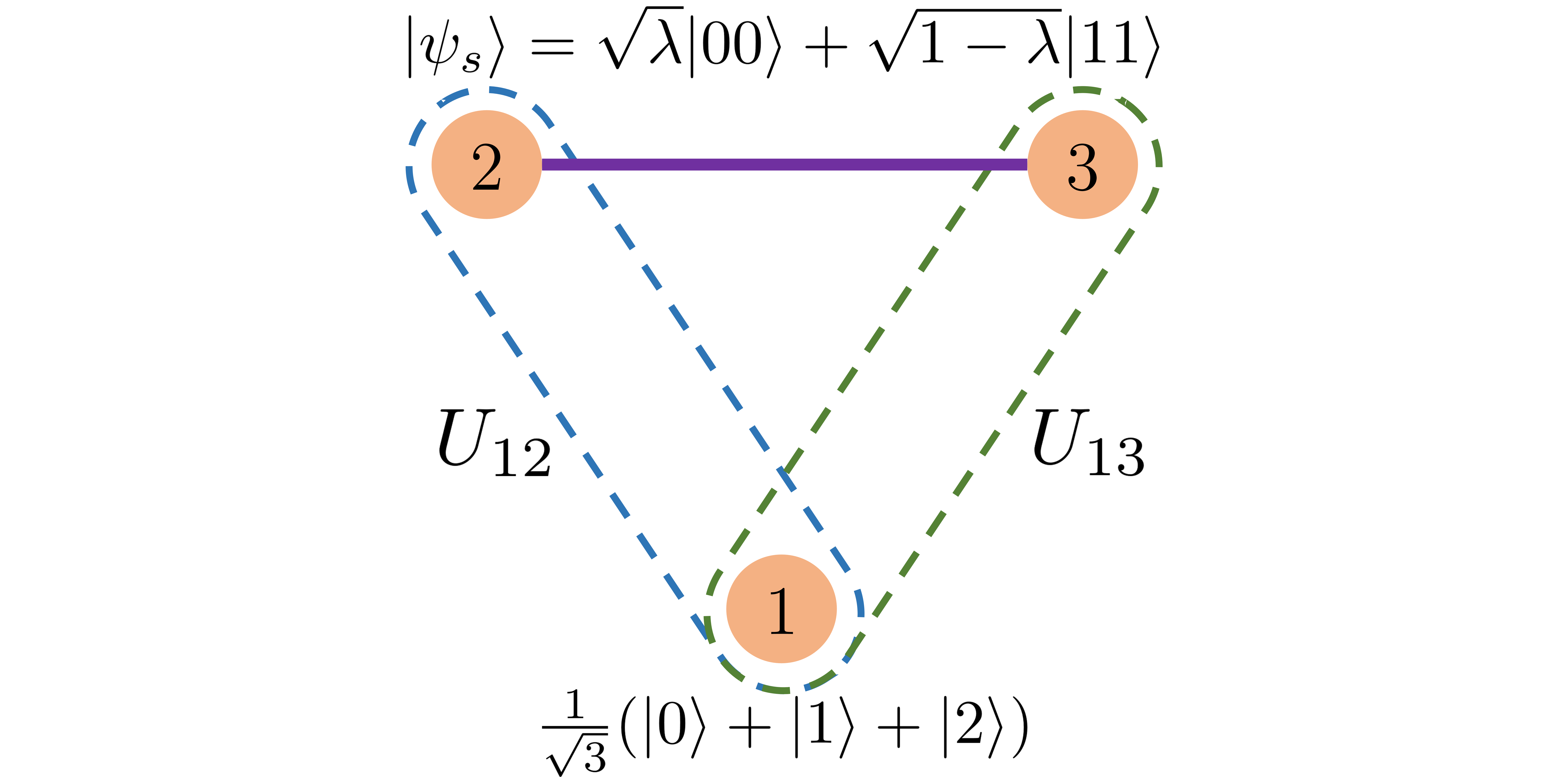}
    \caption{Any state in $M_{A_3}$ can be prepared (up to LU) by entangling qutrit 1 in the state $\frac{1}{\sqrt{3}}(\ket{0}+\ket{1}+\ket{2})$ to the entangled qubits 2 and 3 in the state $\ket{\psi_s}$ with 2-qutrit entangling unitaries $U_{12}$ and $U_{13}$.}\label{Fig:Decomp_MES_A3}
\end{figure}

\section{Conclusion}
Although isolation is a generic feature of multipartite states in $(\mathbb{C}^d)^{\otimes n}$ for $n>3$ and $d>2$ (and for $n>4$ when $d=2$), there could still exist zero-measure SLOCC classes with a rich LOCC structure. In fact, such examples are known (most notably the $n$-qubit GHZ and W SLOCC classes) but they are scarce. Thus, in this work we considered this question in general with the hope of finding more isolation-free SLOCC classes. In order to do so, we investigated whether weak isolation exists in many SLOCC classes. First, we proved that weakly isolated states are present in all $n$-qudit SLOCC classes that have fully entangled states with compact stabilizers for arbitrary local dimensions. Then, we considered $n$-qudit SLOCC classes that has a fully entangled state stabilized by non-compact stabilizers restricted to the form $S^{\otimes n}$ for $S\in GL(d,\mathbb{C})$ and any local dimension $d$. We proved that for $n=3$, the SLOCC class is isolation-free for any $d\geq2$. We also showed that the SLOCC class of the 3-qutrit totally antisymmetric state $\ket{A_3}$ is a special case with $d=3$, thereby identifying a new isolation-free class. However, we proved that weakly isolated states exist in $n$-qudit SLOCC classes with non-compact stabilizers of the form $S^{\otimes n}$ for all $n\geq4$ and $d\geq2$. Examples of such classes are those of $n$-qunit totally antisymmetric states for all $n\geq4$. Let us mention here that the $\ket{A_3}$ class is singled out by the fact that this state and the 2-qubit singlet state are the only absolutely maximally entangled states, i.e., states which are maximally entangled in any bipartite splitting, in $(\mathbb{C}^d)^{\otimes n}$ for all $d,n\geq2$ that satisfy $U^{\otimes n}\ket{\psi}\propto \ket{\psi}$ for any $U\in U(d,\mathbb{C})$ \cite{A3vsAngeq4}. However, it is still unclear whether these properties are relevant for the absence of isolation.

Finally, we studied the entanglement properties of the new isolation-free class. We proved that every state in such class can be reached (up to LU) from a zero-measure subset of states, $M_{A_3}$, with LOCC using at most 3 rounds, while the states in $M_{A_3}$ are not $\text{LOCC}_\mathbb{N}$-reachable from any other states. Moreover, we showed that all states in $M_{A_3}$ can be prepared with a simple protocol using only bipartite entanglement resources.

The reason for asking the initial question was that the elements of the MES of the $n$-qubit GHZ and W SLOCC classes appear in many practical applications \cite{SecretSharing,ByzantineGHZ,QCommW,MetrologyGHZ,RobustLossW}. Thus, arguably, one might expect to find useful states in other isolation-free classes due to their exceptionally rich entanglement structure. Indeed, the newly identified isolation-free class proven in this work contains the state $|A_3\rangle$ in its MES, which also plays an important role in many known practical protocols \cite{AdanAntiSymmAppl,ByzantineAgree,InvertUnitary,InvertIsometry}. In future work, it would be interesting to study whether a practical scenario could be found, in which the very notion of ``isolation-freeness'' translates into some form of advantage. On the more technical level, it would be desirable to identify more isolation-free classes or to prove that the examples found thus far exhaust all possibilities. In this regard, it would be particularly useful to devise new techniques that would make it possible to prove that weakly isolated SLOCC classes are in fact isolated or to find a counterexample.

\begin{acknowledgements}

We thank Nolan Wallach and Tobias Fritz for helpful discussions regarding Proposition \ref{prop:Compact-Finite}.

NKHL, CS, MH, and BK acknowledge financial support from the Austrian Science Fund (FWF): W1259-N27 (DK-ALM), F 7107-N38 (SFB BeyondC) and P 32273-N27 (Stand-Alone Project). JIdV acknowledges financial support from the Spanish Ministerio de Ciencia e Innovaci\'on (grant PID2020-113523GB-I00 and ``Severo Ochoa Programme for Centres of Excellence'' grant CEX2019-000904-S funded by MCIN/AEI/10.13039/501100011033) and from Comunidad de Madrid (grant QUITEMAD-CM P2018/TCS-4342 and the Multiannual Agreement with UC3M in the line of Excellence of University Professors EPUC3M23 in the context of the V PRICIT)

\end{acknowledgements}

\appendix

\section{Proof of Observation \ref{obs:An_symm_group}}\label{app:An_symm}

In this appendix, we prove Observation \ref{obs:An_symm_group} which states the symmetry group of $n$-qunit totally antisymmetric states for all $n\geq2$.

\begin{repobservation}{obs:An_symm_group}
The symmetry group of the $n$-qunit totally antisymmetric state $\ket{A_n}$ is $\{S^{\otimes n} \mbox{ with } S\in SL(n,\mathbb{C})\}$.
\end{repobservation}
\begin{proof}

Using Eq.\;(\ref{eq:An_perm}), it is straightforward to see that the existence of a symmetry, which is not of the form $S^{\otimes n}$ implies that there exists a $B\in SL(n,\mathbb{C})$ such that \cite{footnote:BBinv}
\begin{align}
    B\otimes B^{-1}\otimes \mathbbm{1}^{\otimes n-2}\ket{A_n}=\ket{A_n} \label{eq:TotAntiSymm_BBinv}.
\end{align}
It follows from Eq.\;(\ref{eq:TotAntiSymm_BBinv}) that
\begin{align}
    (B\otimes \mathbbm{1}-\mathbbm{1}\otimes B)\otimes\mathbbm{1}^{\otimes n-2}\ket{A_n}=0. \label{eq:TotAntiSymm_B1-1B}
\end{align}
Let $B=\sum_{i,j=1}^{n} B_{i,j}|i\rangle\langle j|$. Projecting the last $n-2$ systems of the state in Eq.\;(\ref{eq:TotAntiSymm_B1-1B}) onto $\ket{i_3,\ldots,i_n}$ gives
\begin{align}
    (B\otimes \mathbbm{1}-\mathbbm{1}\otimes B)(\ket{i_1,i_2}-\ket{i_2,i_1}) = 0
\end{align}
for all $i_1 \neq i_2$. Its overlaps with $\bra{i_1,i_2}$ and $\bra{i_1,i_1}$ give the constraints $B_{i_1, i_1}=B_{i_2, i_2}$ and $B_{i_1, i_2}=0$, respectively, for all $i_1\neq i_2$. Thus, only $B\propto\mathbbm{1}$ can satisfy Eq.\;(\ref{eq:TotAntiSymm_BBinv}) and hence, any symmetry of $\ket{A_n}$ must be equal to $S^{\otimes n}$ for some $S\in GL(n,\mathbb{C})$.

The remaining part of the proof is to show that for any $S\in GL(n,\mathbb{C})$, $\frac{1}{\det S}S^{\otimes n}$ is a symmetry of $\ket{A_n}$. Consider
\begin{align}
    &\langle j_1,\ldots,j_n|S^{\otimes n}|A_n\rangle \\
    = &\;\frac{1}{\sqrt{n!}}\sum_{i_1,\ldots,i_n=1}^n \epsilon_{i_1,\ldots,i_n}[S]_{j_1, i_1}[S]_{j_2, i_2}\ldots [S]_{j_n, i_n}\\
    = &\;\frac{1}{\sqrt{n!}}\epsilon_{j_1,\ldots,j_n} \det S \label{eq:m_Sn_An}
\end{align}
where the last line is obtained from the Leibniz formula for determinants that uses the Levi-Civita symbol. Due to Eq.\;(\ref{eq:An_epsilon}), it follows that $S^{\otimes n}\ket{A_n} = \det S\ket{A_n}$ for all $S\in GL(n,\mathbb{C})$ and therefore,  $S'=S/(\det S)^{\frac{1}{n}} \in SL(n,\mathbb{C}) $ fulfills $S'^{\otimes n}\ket{A_n} = \ket{A_n}$. Thus, the symmetry group of $\ket{A_n}$ is $\{S^{\otimes n} \mbox{ with } S\in SL(n,\mathbb{C})\}$.
\end{proof}

\section{Proof of Corollary \ref{cor:UnitarySymm_WI}}\label{app:proof2ThmU_Symm}

In this appendix, we first prove Corollary \ref{cor:UnitarySymm_WI} using Proposition \ref{prop:USymm-Compact} and then we provide an alternative proof of Corollary \ref{cor:UnitarySymm_WI} for $n$ qubits, which does not require Proposition \ref{prop:Compact-Finite}. Let us first restate Corollary \ref{cor:UnitarySymm_WI} here.

\begin{repcorollary}{cor:UnitarySymm_WI}
Every $n$-qudit SLOCC class that contains a fully entangled state $\ket{\psi}$ which has a unitary symmetry group $\mathcal{S}_\psi \subseteq \bigotimes_{i=1}^n U(d_i,\mathbb{C})$ contains weakly isolated states.
\end{repcorollary}

The proof follows immediately from Theorem \ref{thm:WI_CompactSymm}, which states that weakly isolated states exist in every $n$-qudit SLOCC class that contains a fully entangled state with a compact symmetry group, and the following proposition.

\begin{proposition}\label{prop:USymm-Compact}
Let $\mathcal{S}_\psi\subseteq \bigotimes_{i=1}^n U(d_i,\mathbb{C})$ be the symmetry group of a state $\ket{\psi}$. Then, $\mathcal{S}_\psi$ is compact.
\end{proposition}

\begin{proof}
The symmetry group of a state $\ket{\psi}$, $\mathcal{S}_\psi=\{S\in \bigotimes_{i=1}^n U(d_i,\mathbb{C}): (S-\identity)\ket{\psi}=0\}$, is the intersection of the compact group $\bigotimes_{i=1}^n U(d_i,\mathbb{C})$ and a Zariski closed set, namely the set of all solutions to the finite set of complex polynomial equations given above. As Zariski closed sets are closed in the Euclidean topology \cite{footnote:ZariskiEuclidean} and  $\mathcal{S}_\psi$ is bounded since it is a subgroup of a compact group, we have that $\mathcal{S}_\psi$ is also compact.
\end{proof}

For the alternative proof restricted to $n$-qubit cases, we will need the following lemmata. Our first step is to show that every local symmetry group corresponding to a unitary symmetry group must be closed. Note that considering the local symmetry group on a single site, $j$, one can always choose its elements in $SU(d_j,\mathbb{C})$, i.e., $\mathcal{S}_{\psi}^{(j)}=\{S^{(j)}\in SU(d_j,\mathbb{C}):\exists\;S^{(i)}\;\forall\;i\in[n]\setminus\{j\} \text{ so that } \bigotimes_{k=1}^n S^{(k)}\in\mathcal{S}_\psi\}$. Then, we have the following lemma.

\begin{lemma}\label{lemma:closedSymm}
Let $\mathcal{S}_\psi\subseteq \bigotimes_{i=1}^n U(d_i,\mathbb{C})$ be the symmetry group of a state $\ket{\psi}$. Then, the local symmetry group $\mathcal{S}_{\psi}^{(j)}\in SU(d_j,\mathbb{C})$ is closed for any choice of $j\in[n]$.
\end{lemma}

\begin{proof}
We will prove the statement by contradiction. By Proposition \ref{prop:USymm-Compact}, $\mathcal{S}_\psi$ is compact. Assume that $\mathcal{S}_\psi^{(j)}$ is not closed for some $j$. Then there exists a convergent sequence $\{S_k^{(j)}\}_{k\in\mathbb{N}}\subseteq \mathcal{S}_\psi^{(j)}$ such that $\lim_{k\rightarrow\infty}S_k^{(j)}\not\in \mathcal{S}_\psi^{(j)}$. As $\mathcal{S}_\psi$ is compact, the corresponding sequence $\{\bigotimes_{l=1}^n S_k^{(l)}\}_{k\in\mathbb{N}}\subseteq\mathcal{S}_\psi$ is bounded. By the Bolzano-Weierstrass theorem \cite{BolzanoWeierstrass,footnote:BW}, there exists a convergent subsequence $\{\bigotimes_{l=1}^n \tilde{S}_k^{(l)}\}_{k\in\mathbb{N}}\subset\{\bigotimes_{l=1}^n S_k^{(l)}\}_{k\in\mathbb{N}}$ and its limit must be an element in $\mathcal{S}_\psi$ since this set is closed. Here, we choose the phase of $\tilde{S}_k^{(j)}$ such that $\tilde{S}_k^{(j)}$ is an element of $\{S_k^{(j)}\}_{k\in\mathbb{N}}$ for each $k$. Since the limit of a convergent sequence ($\{S_k^{(j)}\}_{k\in\mathbb{N}}$) equals to the limit of any  subsequence, we have $\lim_{k\rightarrow\infty}\tilde{S}_k^{(j)}=\lim_{k\rightarrow\infty}S_k^{(j)}$, which leads to a contradiction as the LHS is an element of $ \mathcal{S}_\psi^{(j)}$, whereas the RHS is not. Therefore, $\mathcal{S}_\psi^{(j)}$ is closed for any choice of $j\in[n]$.
\end{proof}

Due to the argument in the proof of Theorem \ref{thm:WI_CompactSymm}, it remains to consider SLOCC classes that have a state $\ket{\psi}$ with an infinite symmetry group $\mathcal{S}_\psi$. A simple counting argument shows that $\mathcal{S}_\psi$ is infinite if and only if the local symmetry group $\mathcal{S}_\psi^{(j)}$ is infinite for at least one $j\in[n]$. Due to  Lemma \ref{lemma:closedSymm} and the fact that the only infinite closed subgroups of $SU(2,\mathbb{C})$ are up to unitary conjugation $\mathcal{G}_1=\{diag(e^{i\theta},e^{-i\theta})\}_{\theta\in[0,2\pi)}$, $\mathcal{G}_2=\mathcal{G}_1\cup\left\{\begin{pmatrix} 0 & -e^{i\varphi}\\ e^{-i\varphi} & 0\end{pmatrix}\right\}_{\varphi\in[0,2\pi)}$, and $SU(2,\mathbb{C})$ (see Sec.\;III of Ref.\;\cite{ClosedSubgroupSU2}), we deduce that $\mathcal{S}_\psi^{(j)}$ is, up to unitary conjugation, one of these three subgroups of $SU(2,\mathbb{C})$ for at least one $j$.

With the help of the subsequent lemma, we will now show that there exists no SLOCC class with an infinite unitary local stabilizer. In particular, we show that any of these local symmetry groups must also contain matrices in $SL(2,\mathbb{C})\setminus SU(2,\mathbb{C})$. Hence, there exists no SLOCC class with unitary stabilizer without weak isolation. Without loss of generality, we can assume that $\mathcal{S}_\psi^{(1)}\in\{\mathcal{G}_1,\mathcal{G}_2,SU(2,\mathbb{C})\}$. Since $\mathcal{G}_1\subset\mathcal{G}_2\subset SU(2,\mathbb{C})$, it suffices to consider unitary symmetries $\bigotimes_{i=1}^n U_i$ with $U_1=diag(e^{i\theta},e^{-i\theta})$ for all $\theta\in[0,2\pi)$ for the rest of the proof \cite{footnote:LUequiv}.

\begin{lemma}\label{lemma:SU2toSL2}
If an $n$-qubit state $\ket{\psi}$ has a unitary symmetry $U(\theta)=e^{-i\varphi}diag(e^{i\theta},e^{-i\theta})\otimes_{i=2}^n U_i\in U(2,\mathbb{C})^{\otimes n}$ for $\theta=x\pi$ where $x\in(0,2)\cap(\mathbb{R}\setminus\mathbb{Q})$, then $\ket{\psi}$ also has a symmetry $X(z) = w(z) diag(e^{z\theta},e^{-z\theta})\otimes_{i=2}^n X_i(z)$ where $w(z)\in\mathbb{C}$ and $X_i(z)\in SL(2,\mathbb{C})$ for every $z\in\mathbb{C}$.
\end{lemma}

Let us first sketch the proof of Lemma \ref{lemma:SU2toSL2} here and then give the complete proof once we prove Lemma \ref{lemma:vecZ0consistent}. The first step is to consider the spectral decomposition of $U_j$, which is a tensor factor of the unitary $U(\theta)\in U(2,\mathbb{C})^{\otimes n}$, as $U_j=e^{i\alpha_j}|v_j^{(0)}\rangle\langle v_j^{(0)}|+e^{-i\alpha_j}|v_j^{(1)}\rangle\langle v_j^{(1)}|$ and determine how $\{\alpha_j\}_{j=2}^n$ depend on $\varphi$ and $\theta$ for $U(\theta)$ to be a symmetry of $\ket{\psi}$. Then, we construct the operators $X(z)\in GL(2,\mathbb{C})^{\otimes n}$ with $w(z)=e^{-z\varphi'}$ and $X_j(z)=e^{z\alpha'_j}|v_j^{(0)}\rangle\langle v_j^{(0)}|+e^{-z\alpha'_j}|v_j^{(1)}\rangle\langle v_j^{(1)}|$ where $z\in\mathbb{C}\setminus i\mathbb{R}$ and $(\varphi',\{\alpha'_j\}_{j=2}^n)$ are related to $(\varphi,\{\alpha_j\}_{j=2}^n)$ by being solutions to two systems of linear equations that share the same coefficient matrix (see Lemma \ref{lemma:vecZ0consistent}). Finally, we show that $X(z)$ is a symmetry of $\ket{\psi}$ by using the fact that $(\varphi',\{\alpha'_j\}_{j=2}^n)$ solves the system of linear equations which corresponds to $X(z)\ket{\psi}=\ket{\psi}$. Note that we will use $\vec{1}_{d}$ to denote a $d$-dimensional vector with all entries being $1$.

\begin{lemma}\label{lemma:vecZ0consistent}
Let $\theta=x\pi$ where $x\in(0,2)\cap(\mathbb{R}\setminus\mathbb{Q})$. If there exists a solution $\vec{\alpha}\in\mathbb{R}^L$ to the matrix equation $\mathbf{B}\vec{\alpha}=\vec{\varphi}+\vec{\theta}+2\pi\vec{m}$ where $\mathbf{B}$ is a $(d_1+d_2)\times L$ matrix with rational entries, $\vec{\varphi}=\varphi\cdot\vec{1}_{d_1+d_2}$, $\varphi\in\mathbb{R}$, $\vec{\theta}=\theta\cdot((-\vec{1}_{d_1})\oplus\vec{1}_{d_2})$, and $\vec{m}\in\mathbb{Z}^{d_1+d_2}$, then there exists a choice for $\varphi'\in\mathbb{R}$ such that the equation $\mathbf{B}\vec{\alpha'}=\vec{\varphi'}+\vec{\theta}$ where $\vec{\varphi'}=\varphi'\cdot\vec{1}_{d_1+d_2}$ has a real solution $\vec{\alpha'}\in\mathbb{R}^L$.
\end{lemma}
\begin{proof}
Consider a system of equations represented by the matrix equation $\mathbf{B}\vec{\alpha}=\vec{\varphi}+\vec{\theta}+2\pi\vec{m}$. 
While allowing any free choice of $\varphi\in\mathbb{R}$ and $\vec{m}\in\mathbb{Z}^{d_1+d_2}$, we determine the necessary and sufficient conditions for the system of equations to be consistent (and, hence, admits a real solution). We then choose $\varphi'\in\mathbb{R}$ so that the above conditions imply that the system $\mathbf{B}\vec{\alpha'}=\vec{\varphi'}+\vec{\theta}$ is consistent as well. Notice that a system with a matrix of coefficients given by $\mathbf{B}$ is always consistent if all the rows of this matrix are linearly independent, so we only need to analyze the situation when this is not the case.

By performing Gaussian elimination on the augmented matrix corresponding to $\mathbf{B}\vec{\alpha}=\vec{\varphi}+\vec{\theta}+2\pi\vec{m}$, the system is consistent if and only if for each row of the resulting matrix where the first $L$ entries are all zeros, the last entry must also be zero. We will call these conditions ``consistency conditions''. Given that $\mathbf{B}$ is a rational matrix, they take the form
\begin{equation}
    x_1^{(\lambda)} \varphi +x_2^{(\lambda)}\theta+2\pi\sum_j y_j^{(\lambda)} m_j=0 \label{eq:consistency}
\end{equation}
where $x_1^{(\lambda)},x_2^{(\lambda)},y_j^{(\lambda)}\in\mathbb{Q}$ for all $j\in\{1,\ldots,d_1+d_2\}$, $\lambda$ labels the different consistency conditions, and $m_j$ is the $j$-th entry of $\vec{m}$. Notice that $\theta$ being an irrational multiple of $\pi$ imposes that $f=\frac{x_2^{(\lambda)}}{x_1^{(\lambda)}}\in\mathbb{Q}$ is a constant independent of $\lambda$ for all the consistency conditions for which it does not hold that $x_1^{(\lambda)}=x_2^{(\lambda)}=0$ \cite{footnote:consistencyCond}.

We consider now the system $\mathbf{B}\vec{\alpha'}=\vec{\varphi'}+\vec{\theta}$ and look for a choice of $\varphi'\in\mathbb{R}$ that makes it consistent under the assumption that all the previous consistency conditions (\ref{eq:consistency}) are met. The consistency conditions are now
\begin{equation}
    x_1^{(\lambda)} \varphi' +x_2^{(\lambda)}\theta=0 \label{eq:consistency2}
\end{equation}
where the coefficients $\{x_1^{(\lambda)},x_2^{(\lambda)}\}$ are the same as before. Observe first that if the previous system had consistency conditions of the form $x_1^{(\lambda)}=x_2^{(\lambda)}=0$, these are trivially met here for any choice of $\varphi'$. On the other hand, all remaining conditions in Eq.\ (\ref{eq:consistency2}) are fulfilled by choosing $\varphi'=-f\theta$.
\end{proof}

We are now in the position to prove Lemma \ref{lemma:SU2toSL2}.
\begin{proof}[Proof of Lemma \ref{lemma:SU2toSL2}]
Let $\ket{\psi}=\ket{0}_1\ket{\phi_0}_{2\ldots n}+\ket{1}_1\ket{\phi_1}_{2\ldots n}$ be an $n$-qubit state where $\ket{\phi_0},\ket{\phi_1}\in(\mathbb{C}^2)^{\otimes(n-1)}$ are both unnormalized. Suppose that $U(\theta) = e^{-i\varphi}diag(e^{i\theta},e^{-i\theta})\otimes_{j=2}^n U_j\in U(2,\mathbb{C})^{\otimes n}$ with $U_j=e^{i\alpha_j}|v_j^{(0)}\rangle\langle v_j^{(0)}|+e^{-i\alpha_j}|v_j^{(1)}\rangle\langle v_j^{(1)}|$, where $\theta=x\pi$ with $x\in(0,2)\cap(\mathbb{R}\setminus\mathbb{Q})$, $\varphi,\alpha_j\in\mathbb{R}$, and $\{|v_j^{(k)}\rangle\}_{k=0}^1$ are orthonormal for all $j\in\{2,\ldots,n\}$, is a symmetry of $\ket{\psi}$ (i.e., $U(\theta)\ket{\psi}=\ket{\psi}$). This means that
\begin{equation}
    \bigotimes_{j=2}^n \left(\sum_{k=0}^1 e^{ i(-1)^k\alpha_j}|v_j^{(k)}\rangle\langle v_j^{(k)}|\right)\ket{\phi_l}=e^{i(\varphi-(-1)^l\theta)}\ket{\phi_l}
    \label{eq:eigenEq}
\end{equation}
has to hold for both $l=0$ and 1. These two equations are satisfied if and only if
\begin{equation}
    \ket{\phi_l} = \sum_{\vec{k}\in\{0,1\}^{n-1}} c_{\vec{k}}^{(l)} \;\bigotimes_{j=2}^n |v_j^{(k_j)}\rangle
    \label{eq:decomp_phi_l}
\end{equation}
where $c_{\vec{k}}^{(l)}\in\mathbb{C}$ and $c_{\vec{k}}^{(l)}\neq0$ only for those $\vec{k}$, for which it holds that  $e^{i\sum_{j=2}^n (-1)^{k_j}\alpha_j} = e^{i(\varphi-(-1)^l\theta)}$ for $l=0$ and 1, which is equivalent to
\begin{equation}
    \sum_{j=2}^n (-1)^{k_j}\alpha_j = \varphi-(-1)^l\theta + 2\pi m
    \label{eq:eigenvalueCondition}
\end{equation}
for some $m\in\mathbb{Z}$. The collection of all constraints on $\{\alpha_j\}_{j=2}^n$ [in the form of Eq.\;(\ref{eq:eigenvalueCondition})] associated to all $c_{\vec{k}}^{(l)}\neq0$ for $l=0$ and 1 form a system of linear equations of the form, $\mathbf{B}\vec{\alpha}=\vec{\varphi}+\vec{\theta}+2\pi\vec{m}$ where $\mathbf{B}$ is a $(d_1+d_2)\times (n-1)$ matrix with entries either $1$ or $-1$, $d_{l+1}$ is the number of $c_{\vec{k}}^{(l)}\neq0$, $\vec{\alpha}=(\alpha_2,\ldots,\alpha_n)^T$, and $\vec{\varphi}$, $\vec{\theta}$, and $\vec{m}$ are defined in the same way as in Lemma \ref{lemma:vecZ0consistent}. 

As $U(\theta)$ is unitary, the equation $\mathbf{B}\vec{\alpha}=\vec{\varphi}+\vec{\theta}+2\pi\vec{m}$ must have a real solution. Then, by Lemma \ref{lemma:vecZ0consistent}, $\mathbf{B}\vec{\alpha'}=\vec{\varphi'}+\vec{\theta}$ must also have a real solution $\vec{\alpha'}\in\mathbb{R}^{n-1}$, which implies that \cite{footnote:how_to_get_exp_equal}
\begin{equation}
    e^{\sum_{j=2}^n (-1)^{k_j}z\alpha'_j} = e^{z(\varphi'-(-1)^l\theta)}
    \label{eq:z&ExpPhases}
\end{equation}
must hold for all $z\in\mathbb{C}$, $\vec{k}\in\{0,1\}^{n-1}$ and $l\in\{0,1\}$ such that $c_{\vec{k}}^{(l)}\neq0$. 

Now, let $X(z) = w(z) diag(e^{z\theta},e^{-z\theta})\otimes_{i=2}^n X_i(z)$ where $w(z)=e^{-z\varphi'}$, $X_j(z)=e^{z\alpha'_j}|v_j^{(0)}\rangle\langle v_j^{(0)}|+e^{-z\alpha'_j}|v_j^{(1)}\rangle\langle v_j^{(1)}|$, and $z\in\mathbb{C}$. Note that $X(z)\in U(2,\mathbb{C})^{\otimes n}$ if and only if $Re(z)=0$. Using the same methods as above, it is straightforward to see that $X(z)\ket{\psi}=\ket{\psi}$ for all $z\in\mathbb{C}$. Hence, there also exist non-unitary symmetries, which completes the proof.
\end{proof} 

We are now in the position to present an alternative proof of Corollary \ref{cor:UnitarySymm_WI} for $n$ qubits. From the argument above, SLOCC classes that have a state $\ket{\psi}$ with an infinite symmetry group $\mathcal{S}_\psi \subseteq U(2,\mathbb{C})^{\otimes n}$ means that there exists an LU-equivalent state $\ket{\phi}$ such that $\mathcal{S}_\phi^{(j)}\in\{\mathcal{G}_1,\mathcal{G}_2,SU(2,\mathbb{C})\}$ for at least one $j\in[n]$. We will now prove that no such class exists. Recall that, without loss of generality, we can take $\mathcal{S}_\phi^{(1)}\in\{\mathcal{G}_1,\mathcal{G}_2,SU(2,\mathbb{C})\}$ and $U(\theta)=e^{-i\varphi}diag(e^{i\theta},e^{-i\theta})\otimes_{i=2}^n U_i\in \mathcal{S}_\phi$ where $\theta=x\pi$ and $x\in(0,2)\cap(\mathbb{R}\setminus\mathbb{Q})$. By Lemma \ref{lemma:SU2toSL2}, the symmetry group $\mathcal{S}_\phi$ must also contain $X(z)\in GL(2,\mathbb{C})^{\otimes n}$ for $z\in\mathbb{C}$ such that $Re(z)\neq0$, which are not in $U(2,\mathbb{C})^{\otimes n}$. Hence, these classes with an infinite stabilizer $\mathcal{S}_\phi\subseteq U(2,\mathbb{C})^{\otimes n}$ do not exist. In particular, all $n$-qubit SLOCC classes that have a state $\ket{\psi}$ with a symmetry group $\mathcal{S}_\psi \subseteq U(2,\mathbb{C})^{\otimes n}$ contain weakly isolated states.

\section{Proof of Theorem \ref{thm:Sn_weakIsolation}}\label{app:proof_Sn_NOweakIsolation}

In this appendix, we prove Theorem \ref{thm:Sn_weakIsolation} which is stated in Sec.\;\ref{sec:noncompact}. We restate the theorem here for clarity.

\begin{reptheorem}{thm:Sn_weakIsolation}
Let $\ket{\Psi_s}$ be an $n$-qudit state $(n\geq4)$ with a symmetry group $\mathcal{S}_{\Psi_s}\subseteq\{S^{\otimes n} \mbox{ with } S\in GL(d,\mathbb{C})\}$ with $d\geq 2$ arbitrary. Then, there exist weakly isolated states within the SLOCC class of $\ket{\Psi_s}$.
\end{reptheorem}
\begin{proof}
Suppose that the state $\ket{\Psi_s}$ has symmetries only of the form $S^{\otimes n}$. We will show that a weakly isolated state in such an SLOCC class can be constructed as $\bigotimes_{i=1}^{n-1}g_i\otimes\mathbbm{1}\ket{\Psi_s}$. Here, the matrices $g_i$ are chosen  such that the only matrices $S\in GL(d,\mathbb{C})$ that satisfy $S^\dagger G_i S \propto G_i$ for $n-1$ sites are $S\propto\mathbbm{1}$ (see Lemma \ref{lemma:isolation}). We divide this proof into two cases: (a) $n\geq5$ and (b) $n=4$. Case (a) has been proven in Ref.\;\cite{OurSymmPaper} \cite{footnote:WI_n>=5}. The main idea there was to choose $G_i\not\propto\identity$ for $i=1,2,3$ and $g_{n-1}=\identity$ so that for any $S$ which  quasi-commute with $G_i$ on  $n-1$ sites, it must quasi-commute with at least two $G_i\not\propto\identity$ and at least one $G_j=\identity$. The last condition implies that $S$ is proportional to a unitary. Appropriate choices of the eigenbases of $G_1,G_2,G_3\not\propto\identity$ force $S\propto\identity$. More precisely, $G_1, G_2$ and $G_3$ are chosen to be non-degenerate and each of these operators has an eigenvector that is non-orthogonal to all eigenvectors of the other two operators. The fact that $S$ has to commute with at least two of these operators implies that $S\propto \identity$. 
For case (b), where $n=4$, such a construction is no longer possible, because there exists $S\not\propto\identity$ that can quasi-commute with $G_1\not\propto\identity$ and $G_3=G_4=\identity$. However, as we will show next, there also exist weakly isolated states.

Similar to case (a), we show that the state $g_1\otimes\ldots\otimes g_4\ket{\Psi_s}$, with $g_i$ defined below is weakly isolated. Let $\{|v^{(1)}_k\rangle=|k\rangle\}_{k=0}^{d-1}$, $\{|v^{(2)}_k\rangle=U_F|k\rangle\}_{k=0}^{d-1}$ and $\{|v^{(3)}_k\rangle=\widetilde{U}U_F|k\rangle\}_{k=0}^{d-1}$ (where $U_F=\frac{1}{\sqrt{d}}\sum_{j,k=0}^{d-1} \omega^{jk}|j\rangle\langle k|$, $\widetilde{U}=|0\rangle\langle 0|+\sum_{k=1}^{d-1} \sqrt{\omega}|k\rangle\langle k|$ and $\omega=e^{i\frac{2\pi}{d}}$) be three orthonormal bases. One can easily verify that they are three mutually non-orthogonal bases for any $d\geq2$. We choose $g_i=\sqrt{G_i}$ where $G_1=\sum_{j=0}^{d-1} r^j |j\rangle\langle j|$ for any $0<r<1$, $G_i=\sum_{j=0}^{d-1} (1-\epsilon)^j|v^{(i)}_j\rangle\langle v^{(i)}_j|$ for $i=2,3$ and $0<\epsilon\ll1$, and $G_4=\mathbbm{1}$.

To prove that the state is weakly isolated, we must prove that no $S\not\propto\identity$ can satisfy $S^\dagger G_i S \propto G_i$ for any three sites. We distinguish the cases: (i) site 4 and any two sites from $\{1,2,3\}$ and (ii) sites 1,2 and 3. In case (i), the proof in case (a) holds directly since the $G_i$ defined here correspond to special instances of the weakly isolated states constructed for $n\geq5$ (by ignoring $G_{4\leq i \leq n-1}$ in the construction for $n\geq5$). In case (ii), we use Observation 1 in Ref.\;\cite{OurSymmPaper} and $S^\dagger G_1 S \propto G_1$ to get $S\propto\sqrt{G_1}^{-1}U\sqrt{G_1}$ for some unitary $U$. Substituting this expression for $S$ into $S^\dagger G_i S \propto G_i$ and using that traces of both sides of this equation must coincide, we obtain $[\sqrt{G_1}^{-1} G_i \sqrt{G_1}^{-1}, U]=0$. We now define $Y=\sqrt{G_1}^{-1} G_2 \sqrt{G_1}^{-1}$ and $Z=\sqrt{G_1}^{-1} G_3 \sqrt{G_1}^{-1}$ which are both positive definite. To prove that only $U\propto\identity$ can satisfy both $[Y,U]=0$ and $[Z,U]=0$, we follow the same argument for proving $S\propto\mathbbm{1}$ if $[S,G_i]=0$ for $i=1,2$ in case (a). It remains to show that (1) $Y$ and $Z$ have non-degenerate eigenvalues and (2) one of the eigenvectors of $Y$ is non-orthogonal to all eigenvectors of $Z$. Despite the proof being straightforward, we present it here for completeness. Since both $G_1$ and $\widetilde{U}$ are diagonal and $G_3=\widetilde{U}G_2 \widetilde{U}^\dagger$, we have that $Z=\widetilde{U}Y\widetilde{U}^\dagger$. In particular, condition (1) is satisfied if $Y$ is non-degenerate. To prove that $Y$ and $Z$ satisfy conditions (1) and (2), we begin by writing $Y$ in the computational basis as
\begin{align}
    Y &= \frac{1}{d}\sum_{j,k,m=0}^{d-1} \frac{(1-\epsilon)^m \omega^{m(j-k)}}{\sqrt{r^{j+k}}}|j\rangle\langle k|\nonumber\\
    &= \frac{1-(1-\epsilon)^d}{d}\sum_{j,k=0}^{d-1} \frac{1}{\sqrt{r^{j+k}}(1-(1-\epsilon)\omega^{(j-k)})}|j\rangle\langle k|\nonumber\\
    &\propto \sum_{j,k=0}^{d-1} \frac{\epsilon}{\sqrt{r^{j+k}}(1-(1-\epsilon)\omega^{(j-k)})}|j\rangle\langle k| \label{eq:Ypropfactor}\\
    &\equiv H_0 + \epsilon V(\epsilon)
\end{align}
where we use the geometric series in the second line to sum over the index $m$ and define a diagonal matrix $H_0=\sum_{k=0}^{d-1}\frac{1}{r^k}|k\rangle\langle k|$ and an off-diagonal matrix $V(\epsilon)=\sum_{j\neq k} \frac{1}{\sqrt{r^{j+k}}(1-(1-\epsilon)\omega^{(j-k)})}|j\rangle\langle k|$. Note that the proportionality factor $\frac{1-(1-\epsilon)^d}{\epsilon d}$ omitted in Eq.\;(\ref{eq:Ypropfactor}) is well defined for $\epsilon>0$. By setting $\epsilon$ sufficiently small, we can use non-degenerate perturbation theory \cite{Sakurai} to find the eigenvalues and eigenvectors of $H_0 + \epsilon V(\epsilon) \propto Y$. 

Since $V(\epsilon)$ depends on the perturbation parameter $\epsilon$, we need to Taylor expand each matrix element around $\epsilon=0$, which results in $V(\epsilon) = \sum_{m=0}^\infty \epsilon^m V^{(m)}$ where
\begin{align}
    V^{(m)} = \sum_{j\neq k} \frac{(-1)^m \omega^{m(j-k)}}{\sqrt{r^{j+k}}(1-\omega^{j-k})^{m+1}}|j\rangle\langle k|.
\end{align}
The solution to the  eigenvalue equation
\begin{align}
    (H_0 + \epsilon V(\epsilon))|e_p\rangle = E_p |e_p\rangle \label{eq:eigeneq}
\end{align}
with $E_p = \sum_{m=0}^\infty \epsilon^m E^{(m)}_p$ and $|e_p\rangle=\sum_{m=0}^\infty \epsilon^m |p^{(m)}\rangle$ \cite{footnote:pert_converge} is given (up to first order) by $E^{(0)}_p=\frac{1}{r^p}$, $|p^{(0)}\rangle = |p\rangle$, $E^{(1)}_p=\langle p^{(0)}|V^{(0)}|p^{(0)}\rangle=0$, and 
\begin{align}
    |p^{(1)}\rangle = \sum_{j\neq p}\frac{\sqrt{r}^{j+p}}{(r^j-r^p)(1-\omega^{j-p})}\ket{j}
\end{align}
for $p\in\{0,\ldots,d-1\}$. Thus, the eigenvalues $E_p = \frac{1}{r^p} + \mathcal{O}(\epsilon^2)$ are non-degenerate if $0<\epsilon\ll1$ is sufficiently small \cite{footnote:pert}. This shows that the eigenvalues of $Y\propto H_0 + \epsilon V(\epsilon)$ and $Z$ are also non-degenerate [satisfying condition (1)].

It remains to show that one of the eigenvectors of $Y$ is non-orthogonal to all eigenvectors of $Z$. Using the fact that $Z=\widetilde{U}Y\widetilde{U}^\dagger$ and $|e_p\rangle$ are the eigenvectors of $Y$, the task becomes proving that for at least one $q$, $\langle e_p|\widetilde{U}|e_q\rangle\neq0$ for all $p\in\{0,\ldots,d-1\}$. Let $|e_p\rangle=\sum_{j=0}^{d-1} e^{(j)}_p|j\rangle$, then
\begin{align}
    \langle e_p|\widetilde{U}|e_q\rangle &= e^{(0)*}_p e^{(0)}_q + \sqrt{\omega}\sum_{j=1}^{d-1} e^{(j)*}_p e^{(j)}_q\label{eq:epUeq2}\\
    &= (1-\sqrt{\omega})e^{(0)*}_p e^{(0)}_q + \sqrt{\omega}\delta_{p,q}.
\end{align}
For $p=q$, this overlap is non-vanishing as $|e^{(0)}_p|^2$ is real and $\frac{-\sqrt{\omega}}{1-\sqrt{\omega}}$ is complex. For all $p\neq q$, $e^{(0)*}_p e^{(0)}_q\neq0$ if and only if $e^{(0)}_p=\langle0|e_p\rangle\neq0 \;\forall\; p$. From perturbation theory (up to first order), we find that 
\begin{align}
    \langle0|e_p\rangle 
    =\begin{cases}1+\mathcal{O}(\epsilon^2), \qquad\qquad\qquad\;\text{if }p=0,\\
    \frac{\epsilon\sqrt{r}^{p}}{(1-r^p)(1-\omega^{-p})} +\mathcal{O}(\epsilon^2), \quad\text{if }p\neq0
    \end{cases}
\end{align}
which are non-zero for all $p\in\{0,\ldots,d-1\}$ if $0<\epsilon\ll1$ is sufficiently small while keeping the eigenvalues $\{E_p\}$ non-degenerate \cite{footnote:Ep_non-deg}. Hence, by choosing $\epsilon$ to be sufficiently small, conditions (1) and (2) are satisfied, so only $U\propto\identity$ can satisfy $[Y,U]=[Z,U]=0$ $\Rightarrow S\propto\sqrt{G_1}^{-1}U\sqrt{G_1}\propto\identity$, which completes the proof for case (b).
\end{proof}

\section{Proof of Observation \ref{obs:nleq4_WI} for $n=4$}\label{app:4qubit_WI}

In this appendix, we will demonstrate the proof of Observation \ref{obs:nleq4_WI} for 4 qubits with one or two selected cases that have an infinite stabilizer from each of the following three families of SLOCC classes. The proofs for all the remaining infinite-stabilizer cases of the following three and all other families of SLOCC classes work analogously and will be omitted here. The full list of SLOCC classes with infinite stabilizers and the exact form of their stabilizers can be found in Ref.\;\cite{4qubitMES}. Note that we will adopt the notation used in Ref.\;\cite{4qubitMES} for labelling different families of SLOCC classes and will neglect state normalization.

The $G_{abcd}$ SLOCC classes have representative/seed states $\ket{\Psi_{abcd}}=a\ket{\Phi^+}\ket{\Phi^+}+b\ket{\Psi^+}\ket{\Psi^+}+c\ket{\Psi^-}\ket{\Psi^-}+d\ket{\Phi^-}\ket{\Phi^-}$ parametrized by $a,b,c,d\in\mathbb{C}$, where $\ket{\Phi^\pm}=\frac{1}{\sqrt{2}}(\ket{00}\pm\ket{11})$ and $\ket{\Psi^\pm}=\frac{1}{\sqrt{2}}(\ket{01}\pm\ket{10})$.
\begin{enumerate}[(a)]
    \item For $a^2=b^2=d^2\neq c^2$ and $a\neq0$, the stabilizers of these seed states are $\{X^{\otimes4}:X\in SL(2,\mathbb{C})\}$, so by Theorem \ref{thm:Sn_weakIsolation}, all these SLOCC classes contain weakly isolated states.
    \item For $a^2=d^2\neq b^2=c^2$ where $a^2=-c^2\neq0$, these SLOCC classes contain states $\ket{\psi_{G_{abcd}}}$ with the stabilizer $\{\langle\sigma_x^{\otimes4},-\sigma_x P_{\sqrt{\frac{y_1}{x_1}}}\otimes P_{\sqrt{\frac{y_2}{x_2}}}\otimes \sigma_x P_{i\sqrt{\frac{x_1}{y_1}}}\otimes P_{i\sqrt{\frac{x_2}{y_2}}}, P_{\sqrt{\frac{z_1}{\eta_1}}}\otimes P_{\sqrt{\frac{z_2}{\eta_2}}}\otimes P_{\sqrt{\frac{\eta_1}{z_1}}}\otimes P_{\sqrt{\frac{\eta_2}{z_2}}}\rangle\}_{x_1,x_2,y_1,y_2,z_1,z_2,\eta_1,\eta_2\in\mathbb{C}\setminus\{0\}}$ where $P_z=\text{diag}(z,z^{-1})$. Let $H_i=\begin{pmatrix} a_i & b_i \\ b_i^* & c_i \end{pmatrix}>0$. Since $\sigma_xP_{z}^\dagger H_i P_{z}\sigma_x\not\propto H_i$ if and only if $z\neq\pm\sqrt{\frac{c_i b_i}{a_i b_i^*}}$, and also $P_{z}^\dagger H_i P_{z}\not\propto H_i$ if and only if $z\neq\pm1$ and $b_i\neq0$, the states $\bigotimes_i h_i\ket{\psi_{G_{abcd}}}$ such that $H_i=h_i^\dagger h_i$ are weakly isolated if $H_i>0$ are non-diagonal for all $i$ and their entries satisfy $|\frac{c_i b_i}{a_i b_i^*}|\neq|\frac{a_j b_j^*}{c_j b_j}|$ for all pairs of indices $(i,j)\in\{(1,3),(2,4)\}$.
\end{enumerate}

The $L_{abc_2}$ SLOCC classes with $c\neq0$ have representative/seed states $\ket{\Psi_{abc_2}}=a\ket{\Phi^+}\ket{\Phi^+}-b\ket{\Phi^-}\ket{\Phi^-}+c(\ket{0110}+\ket{1001})-\frac{i}{2c}\ket{1010}$ parametrized by $a,b,c\in\mathbb{C}$. For $a^2=b^2=c^2$, these SLOCC classes contain states $\ket{\psi_{L_{abc_2}}}$ with the stabilizer $\{T_{z^2,xz}\otimes T_{\frac{1}{z^2},\frac{y}{z}}\otimes T_{z^2,-xz}\otimes T_{\frac{1}{z^2},-\frac{y}{z}}\}_{x,y,z\in\mathbb{C},z\neq0}$ where $T_{x,y}=\begin{pmatrix} 1 & y \\ 0 & x \end{pmatrix}$. For $H_i=\begin{pmatrix} a_i & b_i \\ b_i^* & c_i \end{pmatrix}>0$, $T_{x,y}^\dagger H_i T_{x,y}\propto H_i$ if and only if $y=\frac{b_i}{a_i}(1-x)$ and $|x|=1$. By choosing $H_i$ such that $b_i\neq0$ and $|\frac{b_i}{a_i}|\neq|\frac{b_j}{a_j}|$ for all $(i,j)\in\{(1,3),(2,4)\}$, $\bigotimes_i h_i\ket{\psi_{L_{abc_2}}}$ where $H_i=h_i^\dagger h_i$ are weakly isolated.

The $L_{a_2b_2}$ SLOCC classes with $a,b\neq0$ have seed states $\ket{\Psi_{a_2b_2}}=a(\ket{0011}+\ket{1100})+b(\ket{0110}+\ket{1001})+\frac{i}{2a}\ket{1111}-\frac{i}{2b}\ket{1010}$ parametrized by $a,b\in\mathbb{C}\setminus\{0\}$. For $a^2=b^2$, these SLOCC classes contain states $\ket{\psi_{L_{a_2b_2}}}$ with the stabilizer $\{(\sigma_z\otimes\sigma_y\otimes\sigma_z\otimes\sigma_y)^m\left[T_{1,y}\otimes P_{z^{-1}}\otimes T_{1,-y}\otimes P_z\right]\}_{m\in\{0,1\},y,z\in\mathbb{C}, z\neq0}$. By choosing $H_2$ and $H_4$ non-diagonal, no non-trivial symmetry with $m=0$ can quasi-commute with $\bigotimes_i H_i$ at more than 2 sites. For the remaining symmetries $\bigotimes_i S^{(i)}$ with $m=1$, since the parameters in $S^{(3)}$ and $S^{(4)}$ are completely fixed by $S^{(1)}$ and $S^{(2)}$, once $H_1$ and the non-diagonal $H_2$ are fixed, there always exist $H_3$ and a non-diagonal $H_4$ that cannot be quasi-commuted with $S^{(3)}$ and $S^{(4)}$ of any symmetries with $m=1$. Hence, the states $\bigotimes_i h_i\ket{\psi_{L_{a_2b_2}}}$ where $H_i=h_i^\dagger h_i$ are weakly isolated.

It is straightforward to apply arguments similar to the ones used for the above cases to show that all other remaining 4-qubit infinite-stabilizer SLOCC classes, except the GHZ and W classes, have weakly isolated states.

\section{Proof of Observation \ref{obs:DiagStatesReachAll}}\label{app:DiagStatesReachAll}

In this appendix, we will prove Observation \ref{obs:DiagStatesReachAll} which is stated in Sec.\;\ref{sec:MES_A3}. We restate it here for clarity.

\begin{repobservation}{obs:DiagStatesReachAll}
Any state in the SLOCC class of $\ket{A_3}$ is either LU equivalent to or can be reached from a state in one of the following three sets of states (up to LU) via LOCC in at most 3 rounds:
\begin{enumerate}[(i)]
    \item $\ket{A_3}$,
    \item $\ket{\psi_2(\alpha_1;\alpha_2)}\propto\text{diag}(\sqrt{\alpha_1},1,1)\otimes\text{diag}(\sqrt{\alpha_2},1,1)\otimes\mathbbm{1}\ket{A_3}$ where $\alpha_1,\alpha_2>0$, $\alpha_1\neq\alpha_2$ and $\alpha_1,\alpha_2\neq1$,
    \item $\ket{\psi_3(\alpha_1,\beta_1;\alpha_2,\beta_2)}\propto\text{diag}(\sqrt{\alpha_1},\sqrt{\beta_1},1)\otimes\text{diag}(\sqrt{\alpha_2},\sqrt{\beta_2},1) \otimes \mathbbm{1}\ket{A_3}$ where $\alpha_1,\alpha_2,\beta_1,\beta_2>0$, $\alpha_1\neq\beta_1\neq\beta_2\neq\alpha_2$, $\alpha_1\neq\alpha_2$, $\frac{\alpha_1}{\beta_1}\neq\frac{\alpha_2}{\beta_2}$ and $\alpha_1,\alpha_2,\beta_1,\beta_2\neq1$.
\end{enumerate}
Moreover, these states are not reachable via $\text{LOCC}_\mathbb{N}$. Hence, the MES of $\ket{A_3}$ is contained in the union set of all these states, which we call $M_{A_3}$.
\end{repobservation}

\begin{proof}
Recall from the proof of Theorem \ref{thm:3Qudit_IsoFree} that any state $\ket{\phi}$ in the SLOCC class of $\ket{A_3}$ is equal to $h_1 \otimes h_2 \otimes \mathbbm{1} \ket{A_3}$ with the normalization factor absorbed into $h_1$ or $h_2$. Using the singular value decomposition of $h_2$, $h_2=V\sqrt{D_2}U^\dagger$, where $D_2$ is positive diagonal, and the fact that $U^{\otimes3}$ is a symmetry of $\ket{A_3}$, we have that 
\begin{equation}
    \ket{\phi} \stackrel{\text{LU}}{\cong} h_1 U \otimes \sqrt{D_2} \otimes \mathbbm{1}\ket{A_3}\label{eq:LUequivSLOCC_A3}.
\end{equation}
Any of these states can be reached from a state $\sqrt{D_1}\otimes \sqrt{D_2} \otimes \mathbbm{1}\ket{A_3}$ (up to LU) via $\text{LOCC}_1$, where $D_1$ is positive diagonal. The transformation is achieved if party 1 measures with the operators $\mathcal{K}_{D\rightarrow ND} =  \{\frac{1}{\sqrt{3}}h_1 U \sqrt{D_1}^{-1}Z^j\}_{j=0,1,2}$ and each of the remaining parties apply $Z^j$ depending on the measurement outcome $j$ of party 1, where $Z=\text{diag}(1,e^{i\frac{2\pi}{3}},e^{i\frac{4\pi}{3}})$. One can easily verify that it suffices to take $D_1=\text{diag}(U^\dagger H_1 U)$ so that the operators $\mathcal{K}_{D\rightarrow ND}$ satisfy the completeness relation for a measurement. Therefore, $\ket{\phi}$ is either LU equivalent to or can be reached from
\begin{align}
    &\sqrt{D_1}\otimes \sqrt{D_2} \otimes \mathbbm{1}\ket{A_3}\label{eq:DiagStates}\\
    \propto &\text{\;diag}(\sqrt{\alpha_1},\sqrt{\beta_1},1)\otimes\text{diag}(\sqrt{\alpha_2},\sqrt{\beta_2},1)\otimes\mathbbm{1}\ket{A_3}\nonumber
\end{align}
where $\alpha_i,\beta_i>0$ for $i=1,2$. It remains to show that the three types of states (i)--(iii), which are special cases of Eq.\;(\ref{eq:DiagStates}), are not $\text{LOCC}_\mathbb{N}$-reachable and all the remaining states of the form in Eq.\;(\ref{eq:DiagStates}) can be reached from one of these three types of states with LOCC in 2 steps.

We now prove that states of types (i)--(iii) are not $\text{LOCC}_\mathbb{N}$-reachable by showing that condition (ii) in Theorem \ref{thm:reachability} is violated for all these cases. Clearly, for the type-(i) state ($D_1=D_2=\identity$), any $S\in SL(3,\mathbb{C})$ that quasi-commutes with $\mathbbm{1}$ must quasi-commute with $D_i$ for all sites $i$. For type-(ii) states, any $S\in SL(3,\mathbb{C})$ that quasi-commutes with both $\text{diag}(\alpha,1,1)$ and $\text{diag}(\beta,1,1)$ for $\alpha\neq\beta$ must be of the form $e^{i\theta}\oplus U_2$ where $\theta\in\mathbb{R}$ and $U_2\in U(2,\mathbb{C})$ \cite{footnote:QCwithDiag}. Thus, any $S$ which quasi-commutes with any two operators of the set $\{\text{diag}(\alpha_1,1,1),\text{diag}(\alpha_2,1,1),\identity\}$ must also quasi-commute with the remaining one. Lastly, for type-(iii) states, any $S\in SL(3,\mathbb{C})$ that quasi-commutes with both $\text{diag}(\alpha,\beta,1)$ and $\text{diag}(\gamma,\delta,1)$ where $\frac{\alpha}{\gamma},\frac{\beta}{\delta}$, and 1 take three different values must be unitary and diagonal \cite{footnote:QCwithDiag}. Therefore, any $S$ that quasi-commutes with any two operators of the set $\{\text{diag}(\alpha_1,\beta_1,1),\text{diag}(\alpha_2,\beta_2,1),\identity\}$ must also quasi-commute with the remaining one. Thus, all states of types (i)--(iii) are not $\text{LOCC}_\mathbb{N}$-reachable. Since, by definition, a state is in the MES if and only if it is not reachable via LOCC, the MES of the $\ket{A_3}$ class is contained in the union set of these three types of states.

Finally, we have to show that all the states in the form of Eq.\;(\ref{eq:DiagStates}) but not LU equivalent to type (i), (ii) or (iii) are reachable from one of these three types of states via LOCC within 2 steps. The first case to consider is when $\frac{\alpha_1}{\beta_1}\neq\frac{\alpha_2}{\beta_2}$ and the set $\{\alpha_1,\alpha_2,\beta_1,\beta_2\}$ has at least one element equal to 1. All states in this case are, up to LU (by exchanging local basis $\ket{0}\leftrightarrow\ket{1}$ in all three systems) and systems permutation, equivalent to the state with
\begin{equation}
    D_1=\text{diag}(\alpha, \beta, 1) \text{ and } D_2=\text{diag}(\delta, 1, 1).\label{eq:diagStates1}
\end{equation}
We can further divide this case into the following:
\begin{enumerate}[(a)]
    \item The state $\sqrt{D_1}\otimes \mathbbm{1}^{\otimes2}\ket{A_3}$ (with $\delta=1$) can be obtained from $\ket{A_3}$ [type (i)] with a protocol similar to the one above where party 1 measures with operators $\{\frac{1}{\sqrt{\text{Tr}(D_1)}}\sqrt{D_1}X^j\}_{j=0,1,2}$ with $X=\sum_{i=0}^2|(i+1)\text{mod}\;3\rangle\langle i|$ and parties 2 and 3 each apply $X^j$ depending on the measurement outcome $j$.
    
    \item For $\delta\neq\frac{2\alpha}{\beta+1}$ and $\beta,\delta\neq1$, the state can be obtained from a type-(ii) state $\ket{\psi_2(\alpha_1;\alpha_2)}$ with $\alpha_1=\frac{2\alpha}{\beta+1}$ and $\alpha_2=\delta$ if party 1 measures with operators $\{\frac{1}{\sqrt{\beta+1}}\sqrt{D_1}\tilde{U}^j \text{diag}(\frac{1}{\sqrt{\alpha_1}},1,1)\}_{j=0,1}$ where $\tilde{U}=1\oplus \begin{pmatrix}0& -1 \\
     1 & 0 \end{pmatrix}$ and parties 2 and 3 each apply $\tilde{U}^j$ depending on the measurement outcome $j$.
     
    \item For $\delta=\frac{2\alpha}{\beta+1}$ and $\beta,\delta\neq1$, the state can be obtained from $\ket{A_3}$ in 2 steps: (1) Convert $\ket{A_3}$ into the state $\mathbbm{1}\otimes\mathbbm{1}\otimes\text{diag}(\sqrt{\frac{\beta+1}{2\alpha}},1,1)\ket{A_3} \propto \text{diag}(\sqrt{\frac{2\alpha}{\beta+1}},1,1)^{\otimes2}\otimes\mathbbm{1}\ket{A_3}$ with the protocol in (a) but having the roles of parties 1 and 3 reversed and $D_1$ in the measurement operators replaced by $\text{diag}(\frac{\beta+1}{2\alpha},1,1)$. (2) Party 1 follows protocol (b) to obtain $\sqrt{D_1}\otimes \sqrt{D_2} \otimes \mathbbm{1}\ket{A_3}$.
    
    \item For all other values of $\alpha,\beta,\delta>0$, the state is either a type-(ii) state or can be reached from $\ket{A_3}$ with the protocol in (a) having the roles of parties 1 and 2 (or 3) reversed.
\end{enumerate}

The second case to consider is when $\frac{\alpha_1}{\beta_1}\neq\frac{\alpha_2}{\beta_2}$ and the collection $\{\alpha_1,\alpha_2,\beta_1,\beta_2\}$ has no elements equal to 1. All states belonging to this case are defined by $D_1=\text{diag}(\alpha_1, \beta_1, 1)$ and $D_2=\text{diag}(\alpha_2, \beta_2, 1)$. In fact, they all correspond to the previous case described by Eq.\;(\ref{eq:diagStates1}) up to LU and permutation of systems, except for $\alpha_1\neq\beta_1\neq\beta_2\neq\alpha_2$ and $\alpha_1\neq\alpha_2$, which correspond to the type-(iii) states.

The last case to consider is when $\frac{\alpha_1}{\beta_1}=\frac{\alpha_2}{\beta_2}=\frac{1}{k}$ which corresponds to the states with $D_1=\text{diag}(\alpha, k\alpha, 1)$ and $D_2=\text{diag}(\beta, k\beta, 1)$. Using the symmetry $\frac{\sqrt{D_1}^{\otimes 3}}{\det\sqrt{D_1}}$, it is easy to see that these states are proportional to $\mathbbm{1}\otimes\text{diag}(1,1,\sqrt{\frac{\alpha}{\beta}})\otimes\text{diag}(1, \frac{1}{\sqrt{k}}, \sqrt{\alpha})\ket{A_3}$ which are equivalent to the case described by Eq.\;(\ref{eq:diagStates1}) up to LU and permutation of systems. This completes the proof. \end{proof}

We have shown that every state in the SLOCC class of $\ket{A_3}$ can be reached from a state in $M_{A_3}$ within 3 rounds of LOCC, including the round which converts a state of the form in Eq.\;(\ref{eq:DiagStates}) into the most general state in the SLOCC class described by Eq.\;(\ref{eq:LUequivSLOCC_A3}) with measurement operators $\mathcal{K}_{D\rightarrow ND}$. We observe that certain state transformations [e.g., when case (c) is involved] can be achieved with a non-trivial 3-round protocol. A specific example for that has already been shown in Sec.\;\ref{sec:MES_A3}.

\section{Transformations that are impossible via SEP within $M_{A_3}$}\label{app:SEPforbid}

In this appendix, we will prove Lemma \ref{lemma:MES_LOCCN_SEPforbid} which provides the full details of which states in $M_{A_3}$ cannot be related by SEP (and hence, by LOCC) as summarized in Sec.\;\ref{sec:MES_A3}. At the end, we will define the sets $\mathcal{D}_0$ and $\mathcal{D}_f$ that appear in Fig.\;\ref{Fig:MES_A3}.

\begin{lemma}\label{lemma:MES_LOCCN_SEPforbid}
For any type-(ii) state $\ket{\psi_2(\alpha_1;\alpha_2)}$ in Observation \ref{obs:DiagStatesReachAll}, it holds that: (a) they cannot be reached from any type-(iii) state and (b) they cannot be reached from nor converted to any type-(ii) state $\ket{\psi_2(\alpha'_1;\alpha'_2)}$ with $\alpha'_2\neq\frac{\alpha'_1(\alpha_2-1)+\alpha_1-\alpha_2}{\alpha_1-1}$ by LOCC. Moreover, the state $\ket{A_3}$ is in the MES.
\end{lemma}

\begin{proof}
We will prove that the transformations are not achievable even with SEP by showing that Eq.\;(\ref{eq:SEP}) in Theorem \ref{thm:SEP} cannot be satisfied. Let the initial state be $\sqrt{\Delta'}\otimes \sqrt{D'}\otimes\identity\ket{A_3}$ and the target state be $\sqrt{\Delta}\otimes \sqrt{D}\otimes\identity\ket{A_3}$ where $D, D', \Delta$ and $\Delta'$ are positive diagonal matrices. To obtain a simple necessary condition for the existence of a SEP transformation, we compute the  overlap of Eq.\;(\ref{eq:SEP}) with $\bra{j_1,j_2,j_3}$ and $\ket{A_3}$ after inserting $\Delta \otimes D\otimes \identity$ ($\Delta'\otimes D'\otimes \identity$) for $H$ ($G$), respectively. We get \cite{footnote:SEP_A3_singular}
\begin{equation}
    \sum_{i_1,i_2,i_3,k} p_k \epsilon_{i_1,i_2,i_3}\prod_{m=1}^3 [S_k^\dagger]_{j_m,i_m}\Delta_{i_1}D_{i_2} = r\epsilon_{j_1,j_2,j_3}\Delta'_{j_1}D'_{j_2}.\label{eq:3QutritSEP}
\end{equation}

For case (a), we will show that any type-(iii) state $\ket{\psi_3(\alpha'_1,\beta'_1;\alpha'_2,\beta'_2)}$ cannot be transformed into any type-(ii) state $\ket{\psi_2(\alpha_1;\alpha_2)}$ via SEP. The corresponding Eq.\;(\ref{eq:3QutritSEP}) is equivalent to the following matrix equation
\begin{widetext}
\begin{equation}
    \begin{pmatrix}
    \alpha_1 & -\alpha_1 & -1 & -\alpha_2 & 1 & \alpha_2\\
    -\alpha_1 & \alpha_1 & \alpha_2 & 1 & -\alpha_2 & -1\\
    -\alpha_2 & \alpha_2 & 1 & \alpha_1 & -1 & -\alpha_1\\
    1 & -1 & -\alpha_2 & -\alpha_1 & \alpha_2 & \alpha_1\\
    \alpha_2 & -\alpha_2 & -\alpha_1 & -1 & \alpha_1 & 1\\
    -1 & 1 & \alpha_1 & \alpha_2 & -\alpha_1 & -\alpha_2
    \end{pmatrix}\cdot \vec{x}= r\begin{pmatrix}
    \alpha'_1\beta'_2 \\ -\alpha'_1 \\ -\alpha'_2\beta'_1 \\ \beta'_1 \\ \alpha'_2 \\ -\beta'_2
    \end{pmatrix}\label{eq:type3to2SEP}
\end{equation}
\end{widetext}
where the entries of $\vec{x}$ correspond to the non-vanishing terms of the left hand side of Eq.\;(\ref{eq:3QutritSEP}) and are given by $\sum_{k} p_k \prod_{m=1}^3 [S_k^\dagger]_{j_m,i_m}$ for $i_m,j_m\in\{1,2,3\}$. Using Gaussian elimination, one can easily check that Eq.\;(\ref{eq:type3to2SEP}) has no solution (i.e., inconsistent) if $\alpha_1,\alpha_2,\alpha'_1,\beta'_1,\alpha'_2,\beta'_2$ correspond to the parameters of any type-(ii) and type-(iii) states. Hence, such a transformation is not possible.

For case (b), we will first show that any type-(ii) state $\ket{\psi_2(\alpha'_1;\alpha'_2)}$ cannot be transformed into any type-(ii) state $\ket{\psi_2(\alpha_1;\alpha_2)}$ with $\alpha'_2\neq\frac{\alpha'_1(\alpha_2-1)+\alpha_1-\alpha_2}{\alpha_1-1}$ via SEP. Then, we will show that it is also impossible to transform $\ket{\psi_2(\alpha_1;\alpha_2)}$ into $\ket{\psi_2(\alpha'_1;\alpha'_2)}$ for these parameters with SEP. For the first direction, we simply set $\beta'_1=\beta'_2=1$ in Eq.\;(\ref{eq:type3to2SEP}) and observe that the equation has no solution if
\begin{equation}
    \alpha'_2\neq\frac{\alpha'_1(\alpha_2-1)+\alpha_1-\alpha_2}{\alpha_1-1}.\label{eq:type2to2SEP}
\end{equation}
For the reversed direction, one simply has to exchange the parameters $\alpha_i\leftrightarrow\alpha'_i$ for both $i=1,2$ in Eq.\;(\ref{eq:type3to2SEP}) that has $\beta'_1=\beta'_2=1$ and observe that it also has no solution if Eq.\;(\ref{eq:type2to2SEP}) holds by the fact that Eq.\;(\ref{eq:type2to2SEP}) is symmetric under the exchange, $\alpha_i\leftrightarrow\alpha'_i$.

To show that the type-(i) state, $\ket{A_3}$, is in the MES, one simply sets $\alpha_1=\alpha_2=1$ in Eq.\;(\ref{eq:type3to2SEP}) and sees that the matrix equation is inconsistent when $\alpha'_1,\beta'_1,\alpha'_2,\beta'_2$ correspond to the parameters of any type-(ii) or type-(iii) states. Since $\ket{A_3}$ cannot be reached from any other MES candidates (see Observation \ref{obs:DiagStatesReachAll}) via SEP (LOCC), it must be in the MES \cite{footnote:A3_in_MES}.
\end{proof}

We now define the sets $\mathcal{D}_0$ and $\mathcal{D}_f$ that we mention in Fig.\;\ref{Fig:MES_A3} based on the results in Lemma \ref{lemma:MES_LOCCN_SEPforbid}. We first fix a pair of $(\alpha_1,\alpha_2)$ such that $\alpha_1,\alpha_2>0$, $\alpha_1\neq\alpha_2$, and $\alpha_1,\alpha_2\neq1$. As a zero-measure subset of $M_{A_3}$, $\mathcal{D}_0$ is defined to contain $\ket{A_3}$, a type-(ii) state $\ket{\psi_2(\alpha_1;\alpha_2)}$. The full-measure subset, $\mathcal{D}_f$, of $M_{A_3}$ is defined to contain all type-(iii) states and all type-(ii) states $\ket{\psi_2(\alpha_1';\alpha_2')}$ such that $\alpha'_2\neq\frac{\alpha'_1(\alpha_2-1)+\alpha_1-\alpha_2}{\alpha_1-1}$.

\section{Decomposition of states in $M_{A_3}$}\label{app:decomp_MES_A3}

In Sec.\;\ref{sec:MES_A3}, we stated that any states $\ket{\psi(\alpha_1,\alpha_2,\beta_1,\beta_2)}$ $\propto \text{diag}(\sqrt{\alpha_1},\sqrt{\alpha_2},1)\otimes\text{diag}(\sqrt{\beta_1},\sqrt{\beta_2},1) \otimes \mathbbm{1}\ket{A_3}$ \cite{footnote:ReorderAlpha2Beta1} in $M_{A_3}$ can be prepared with simple bipartite entanglement resources. More precisely, one can show that
\begin{align}
    U_{13}U_{12}\ket{\widetilde{+}}_1\ket{\psi_s}_{23} =V\otimes U_2^\dagger\otimes U_3^\dagger\ket{\psi(\alpha_1,\alpha_2,\beta_1,\beta_2)}
\end{align}
where $\ket{\widetilde{+}}_1=\frac{1}{\sqrt{3}}(\ket{0}+\ket{1}+\ket{2})$, $\ket{\psi_s}_{23}=\sqrt{\lambda_+}\ket{00}+\sqrt{\lambda_-}\ket{11}$, $U_{1j}=\sum_{k=0}^2|k\rangle\langle k|\otimes (U_j^\dagger D_\omega U_j)^k$ for $j=2,3$, $D_\omega=\text{diag}(1,e^{-i\frac{2\pi}{3}},e^{i\frac{2\pi}{3}})$, and $V=\frac{1}{\sqrt{3}}\begin{pmatrix}
1 & 1 & 1\\
1 & \omega & \omega^2\\
1 & \omega^2 & \omega
\end{pmatrix}$ with $\omega=e^{i\frac{2\pi}{3}}$. We now give the exact forms of $\lambda_\pm$, and unitaries $U_2$ and $U_3$ which we omitted in Sec.\;\ref{sec:MES_A3}.

For states $\ket{\psi(\alpha_1,\alpha_2,\beta_1,\beta_2)}$ with $\beta_1\neq\beta_2$, $\lambda_\pm=\frac{1\pm\sqrt{2\chi-1}}{2}$ with $\chi=\frac{\gamma^2-2(\alpha_1+\alpha_2+1)[(\alpha_1+\beta_1)\beta_2+\alpha_2\beta_1]}{\gamma^2}$ and $\gamma=\alpha_1+\alpha_2+\beta_1+\beta_2+\alpha_1\beta_2+\alpha_2\beta_1$. The unitary $U_2=\sigma\begin{pmatrix}
-\frac{\sqrt{\alpha_2\beta_1}x_{+,2}}{\mathcal{N}^{(1)}_+} & -\frac{\sqrt{\alpha_2\beta_1}x_{-,2}}{\mathcal{N}^{(1)}_-} & \frac{1}{\mathcal{N}^{(1)}_0}\sqrt{\frac{\alpha_1}{\beta_1}} \\
\frac{\sqrt{\alpha_1\beta_2}x_{+,1}}{\mathcal{N}^{(1)}_+} & \frac{\sqrt{\alpha_1\beta_2}x_{-,1}}{\mathcal{N}^{(1)}_-} & \frac{1}{\mathcal{N}^{(1)}_0}\sqrt{\frac{\alpha_2}{\beta_2}} \\
\frac{\sqrt{\alpha_1\alpha_2}(\beta_1-\beta_2)}{\mathcal{N}^{(1)}_+} & \frac{\sqrt{\alpha_1\alpha_2}(\beta_1-\beta_2)}{\mathcal{N}^{(1)}_-} & \frac{1}{\mathcal{N}^{(1)}_0} \\
\end{pmatrix}$ with $\sigma=\text{sgn}(\beta_1-\beta_2)$, $x_{\pm,i}=\gamma\lambda_\pm-\alpha_1-\alpha_2-\beta_i$ for $i\in\{1,2\}$, $\mathcal{N}_0^{(m)}=\sqrt{1+\frac{\alpha_1}{\beta_1^{m}}+\frac{\alpha_2}{\beta_2^{m}}}$ and $\mathcal{N}_\pm^{(m)}=\sqrt{\alpha_1\beta_2^{m} x_{\pm,1}^2+\alpha_2\beta_1^{m} x_{\pm,2}^2+\alpha_1\alpha_2(\beta_1-\beta_2)^2}$ and $U_3=\begin{pmatrix}
-\frac{\sqrt{\alpha_2}x_{-,2}}{\mathcal{N}^{(0)}_-} & \frac{\sqrt{\alpha_2}x_{+,2}}{\mathcal{N}^{(0)}_+}
& \frac{\sqrt{\alpha_1}}{\mathcal{N}^{(0)}_0}\\
\frac{\sqrt{\alpha_1}x_{-,1}}{\mathcal{N}^{(0)}_-} & -\frac{\sqrt{\alpha_1}x_{+,1}}{\mathcal{N}^{(0)}_+} & \frac{\sqrt{\alpha_2}}{\mathcal{N}^{(0)}_0}\\
\frac{\sqrt{\alpha_1\alpha_2}(\beta_1-\beta_2)}{\mathcal{N}^{(0)}_-} & -\frac{\sqrt{\alpha_1\alpha_2}(\beta_1-\beta_2)}{\mathcal{N}^{(0)}_+} & \frac{1}{\mathcal{N}^{(0)}_0}
\end{pmatrix}$. 

For any state $\ket{\psi(\alpha_1,\alpha_2,\beta_1,\beta_2)}$ with $\beta_1=\beta_2=\beta$ \cite{footnote:DecompMA3_b1=b2}, the preparation procedure is the same as above, except $\lambda_+=\frac{(\alpha_1+\alpha_2+1)\beta}{\gamma}$, \mbox{$\lambda_-=1-\lambda_+$}, $U_2=\begin{pmatrix}
-\frac{\sqrt{\alpha_2}}{\mathcal{N}'} & -\frac{\sqrt{\alpha_1\beta}}{\mathcal{M}'^{(1)}} & \frac{1}{\mathcal{N}^{(1)}_0}\sqrt{\frac{\alpha_1}{\beta}} \\
\frac{\sqrt{\alpha_1}}{\mathcal{N}'} & -\frac{\sqrt{\alpha_2\beta}}{\mathcal{M}'^{(1)}} & \frac{1}{\mathcal{N}^{(1)}_0}\sqrt{\frac{\alpha_2}{\beta}} \\
0 & \frac{\alpha_1+\alpha_2}{\mathcal{M}'^{(1)}} & \frac{1}{\mathcal{N}^{(1)}_0}
\end{pmatrix}$ and $U_3=\begin{pmatrix}
-\frac{\sqrt{\alpha_1}}{\mathcal{M}'^{(0)}} & \frac{\sqrt{\alpha_2}}{\mathcal{N}'}
& \frac{\sqrt{\alpha_1}}{\mathcal{N}^{(0)}_0}\\
-\frac{\sqrt{\alpha_2}}{\mathcal{M}'^{(0)}} & -\frac{\sqrt{\alpha_1}}{\mathcal{N}'} & \frac{\sqrt{\alpha_2}}{\mathcal{N}^{(0)}_0}\\
\frac{\alpha_1+\alpha_2}{\mathcal{M}'^{(0)}} & 0 & \frac{1}{\mathcal{N}^{(0)}_0}
\end{pmatrix}$ with $\mathcal{N}'=\sqrt{\alpha_1+\alpha_2}$ and $\mathcal{M}'^{(m)}=\mathcal{N}'\sqrt{\alpha_1+\alpha_2+\beta^m}$.


\begin{thebibliography}{99}
\bibitem{footnote:infSLOCC} Notice that Ref.~\cite{4qubitSLOCC} provides 9 families of 4-qubit states but some of these families are collections of an infinite number of inequivalent SLOCC classes (see also, e.g., Chapter 14 in Ref.~\cite{GourBook}).
\bibitem{footnote:NOT1round} Notice that the examples of conversions that cannot be achieved by concatenating one-round protocols do not prove this. This is because the output state is automatically not weakly isolated (it must be finite-round reachable) and the input state can be one-round convertible to a different state.
\bibitem{footnote:LocalSymm_SL} This is because any matrix $\bigotimes_{j=1}^n X^{(j)}\in \bigotimes_{i=1}^n GL(d_i,\mathbb{C})$ is equal to the tensor product between $\frac{X^{(j)}}{\det(X^{(j)})^{1/d_j}}\in SL(d_j,\mathbb{C})$ for any $n-1$ indices $j$ and $\prod_{j\neq k}\det(X^{(j)})^{1/d_j} X^{(k)}$ for the remaining index $k$.
\bibitem{footnote:QuasiComm} We say a matrix $X$ quasi-commutes with another matrix $A$ if and only if $X^\dagger AX= kA\propto A$ for some $k\in\mathbb{C}$.
\bibitem{footnote:PosOpNotCommUs} More precisely, $P$ can be chosen as $P=|v\rangle\langle v|+\identity$, where $\ket{v}\in\mathbb{C}^d$ is not an eigenvector of any $U_i\in\mathcal{F}$. Such a vector always exists as no finite-dimensional vector space over $\mathbb{C}$ is a finite union of proper subspaces (see e.g., Ref.\;\cite{VecSpaceNOTfiniteUnion}).
\bibitem{footnote:MA3_zero_measure} This can be easily seen as follows. First, due to the symmetry of the state, it is easy to see that any state in the SLOCC class is LU equivalent to $\sqrt{G_1}\otimes\sqrt{D_2}\otimes\identity\ket{A_3}$ [see Eq.\;(\ref{eq:LUequivSLOCC_A3})] where $G_1>0$ and $D_2=diag(\alpha_2,\beta_2,1) >0$. Moreover, using the symmetry $U^{\otimes3}$ of $\ket{A_3}$, where $U=diag(e^{i\theta},e^{i\varphi},e^{-i(\theta+\varphi)})$ with $\theta=-\frac{\arg(\gamma_1)+\arg(\delta_1)}{3}$, $\varphi=\frac{2\arg(\gamma_1)-\arg(\delta_1)}{3}$, $\gamma_1=(G_1)_{12}$ and $\delta_1=(G_1)_{13}$, leads to a state of the same form as above, but with $G_1$ replaced by $U G_1 U^\dagger$, whose entries $(1,2)$ and $(1,3)$ are larger than or equal to zero. Hence, the states are (up to LU) parameterized by 8 parameters.
\bibitem{footnote:ReorderAlpha2Beta1} Note that we exchange the order of $\alpha_2$ and $\beta_1$ here as opposed to the notation that we use in Observation \ref{obs:DiagStatesReachAll} to denote the states in $M_{A_3}$.
\bibitem{footnote:BBinv} The argument we use here to show that $B\otimes B^{-1}\otimes\identity^{\otimes n-2}\in\mathcal{S}_{\ket{A_n}}$ is the same argument used in Ref.\;\cite{MigdalSymm} (Sec.\;II) to prove that permutation-symmetric states have symmetries of the form $B\otimes B^{-1}\otimes\identity^{\otimes n-2}$.
\bibitem{footnote:ZariskiEuclidean} See p.8 of Ref.\;\cite{ZariskiClosed} for the fact that Zariski closure on $\mathbb{C}^d$ implies Euclidean closure on $\mathbb{C}^d$.
\bibitem{footnote:BW} It is easy to see that the Bolzano-Weierstrass theorem also applies to bounded sequences in $\mathbb{C}^d$ by viewing them as sequences in $\mathbb{R}^{2d}$.
\bibitem{footnote:LUequiv} The considered state might then be LU-equivalent to the initial state.
\bibitem{footnote:consistencyCond} Note that if there exists a consistency condition with $x_1^{(\lambda)}=0$ and $x_2^{(\lambda)}\neq0$ while $\theta$ is an irrational multiple of $\pi$, then the system of equations is inconsistent.
\bibitem{footnote:how_to_get_exp_equal} We obtain Eq.\;(\ref{eq:z&ExpPhases}) by first multiplying each equation in $\mathbf{B}\vec{\alpha'}=\vec{\varphi'}+\vec{\theta}$ by a factor $z\in\mathbb{C}$ on both sides, and then exponentiating both sides of each equation.
\bibitem{footnote:WI_n>=5} Although the existence of weak isolation was proven for $(n\geq5)$-qudit SLOCC classes of non-exceptionally symmetric (non-ES) states, which are permutation-symmetric states with only symmetries of the form $S^{\otimes n}$, in Lemma 4 of Ref.\;\cite{OurSymmPaper}, the proof also applies to any $n$-qudit SLOCC class that has a state stabilized only by $S^{\otimes n}$ as long as $n\geq5$.
\bibitem{footnote:pert_converge} The perturbation series for $E_p$ and $\ket{e_p}$ are guaranteed to converge because the matrix $H_0 + \epsilon V(\epsilon)$ is Hermitian and analytic (i.e., every matrix entry is analytic) in the neighbourhood of $\epsilon=0$ where $\epsilon\in\mathbb{R}$ and by Rellich's Theorem \cite{Rellich,FriedlandBook}, all the eigenvalues and entries of the eigenvectors must also be analytic in the neighbourhood of $\epsilon=0$.
\bibitem{footnote:pert} Since the perturbation series of eigenvalue $E_p$ converges in $\epsilon$, one can choose $\epsilon$ small enough such that the absolute value of the sum of the $\mathcal{O}(\epsilon^2)$ terms is strictly less than $\frac{1}{2}(\frac{1}{r}-1)$ for $E_0$ and $\frac{1}{2r^{p-1}}(\frac{1}{r}-1)$, which is half the distance between the $(p-1)$-th and the $p$-th unperturbed eigenvalues, for $E_p$ where $p\in\{1,\ldots,d-1\}$ and $0<r<1$.
\bibitem{footnote:Ep_non-deg} Since the perturbation series of eigenvector $\ket{e_p}$ converges in $\epsilon$, one can choose $\epsilon$ small enough such that the absolute value of the sum of the $\mathcal{O}(\epsilon^2)$ terms for $\langle0|e_p\rangle$ is strictly smaller than 1 for $\ket{e_0}$ and $|\frac{\epsilon\sqrt{r}^{p}}{(1-r^p)(1-\omega^{-p})}|$ for every $\ket{e_p}$ where $p\in\{1,\ldots,d-1\}$, while keeping $\{E_p\}$ non-degenerate \cite{footnote:pert}.
\bibitem{footnote:QCwithDiag} It is easy to see the following: If $S\in SL(d,\mathbb{C})$ quasi-commutes with two $d\times d$ positive definite diagonal matrices $\Lambda$ and $D$ such that $\Lambda\not\propto D$, $S$ must be a direct sum of block matrices that act on the (degenerate) eigenspaces of $\Lambda^{-1}D$. Moreover, for each block in $S$ of which the range lies within the (degenerate) eigenspace of a single eigenvalue of $\Lambda$ or $D$, the block is unitary.
\bibitem{footnote:SEP_A3_singular} When multiplying Eq.\;(\ref{eq:SEP}) by $\ket{A_3}$ (which is the seed state $\ket{\Psi_s}$ here) where $g=\sqrt{\Delta'}\otimes \sqrt{D'}\otimes\identity$ and $h=\sqrt{\Delta}\otimes \sqrt{D}\otimes\identity$, the term $g^\dagger\sum_q N_q^\dagger N_q g\ket{A_3}=0$ because all $N_q\in\mathcal{N}_{g\Psi_s}$ satisfy $N_q g\ket{A_3}=0$ by definition.
\bibitem{footnote:A3_in_MES} Alternatively, one can see this by showing that $\ket{A_3}$ is the only state among all the MES candidates in Observation \ref{obs:DiagStatesReachAll} that has a completely mixed single qutrit reduced density matrix for all 3 bipartite splittings. Applying Nielsen's theorem \cite{Nielsen} to all 3 bipartitions proves that $\ket{A_3}$ is indeed not LOCC-reachable.
\bibitem{footnote:DecompMA3_b1=b2} The preparation procedure above does not work for $\ket{\psi(\alpha_1,\alpha_2,\beta_1,\beta_2)}$ with $\beta_1=\beta_2$ because one of the columns in $U_2$ and $U_3$ becomes all zeros when $\beta_1=\beta_2$.
\end{thebibliography}
\end{document}